\documentclass[12pt]{article}
\usepackage{amsmath}
\usepackage{amssymb}
\usepackage{amsthm}
\usepackage{graphicx,psfrag,epsf}
\usepackage{psfrag,epsf}
\usepackage{enumitem}
\usepackage{natbib}
\usepackage{url} 
\usepackage[page]{appendix}

\usepackage{array}
\usepackage{bm}
\usepackage{color}
\usepackage{graphicx}
\usepackage{epstopdf}
\usepackage[letterpaper, portrait, margin=1in]{geometry}
\usepackage{lipsum}
\usepackage{listings}
\usepackage{fancyref}
\usepackage{float}
\usepackage[font=footnotesize,labelfont=bf]{caption}
\usepackage{color}
\usepackage{booktabs}
\usepackage{verbatim}
\usepackage{natbib} 
\usepackage{ifxetex}
\usepackage{bm}
\usepackage{cochineal}
\usepackage{FiraSans}
\usepackage{newtxmath}
\usepackage{anyfontsize}

\newcommand{\z}{\bm{z}}
\newcommand{\Z}{\mathcal{Z}}
\newcommand{\E}{\mathbb{E}}
\renewcommand{\P}{\mathbb{P}}

\usepackage{pifont}
\newcommand{\kp}{k^\prime}

\theoremstyle{plain}
\newtheorem{assumption}{Assumption}
\newtheorem{proposition}{Proposition}

\newtheorem{remark}{Remark}

\usepackage{tikz}
\usetikzlibrary{shapes,arrows,positioning,matrix,calc}
\usepackage{subfig}

\newcommand{\blind}{0}
\newcommand{\submission}{0}
\if0\blind{}
  \def\myauthor{Matthew Blackwell and Nicole Pashley}
\else
\def\myauthor{}
\fi

\usepackage{xr}
\if1\blind{}
\fi

\definecolor{gray}{rgb}{0.459,0.438,0.471}
\definecolor{crimson}{rgb}{0.6,0,0}
\usepackage[plainpages=false, 
            plainpages=false,  
            pdfpagelabels, 
            bookmarksnumbered,
            pdftitle={Bounds on causal effects in 2\string^K factorial experiments with non-compliance}, 
            pdfauthor={\myauthor},
            pdfkeywords={},
            colorlinks=true,
            citecolor=crimson,
            linkcolor=crimson,
            urlcolor=crimson]{hyperref}

\title{\sffamily\bfseries{Bounds on causal effects in $2^{K}$ factorial experiments with non-compliance}%
  \if0\blind{}\thanks{Thanks to Zach Branson for helpful comments. Working paper. Comments welcome.}\fi%
}

\if0\blind{}
\author{Matthew Blackwell\thanks{Department of Government and Institute for Quantitative Social Science, Harvard University. web: \mbox{\url{http://www.mattblackwell.org}} email: \texttt{\href{mailto:mblackwell@gov.harvard.edu}{mblackwell@gov.harvard.edu}}} \and Nicole E. Pashley\thanks{Department of Statistics, Rutgers University. email: \texttt{\href{mailto:nicole.pashley@rutgers.edu}{nicole.pashley@rutgers.edu}}}}
\fi
\date{\today}

\begin{document}

\maketitle

\begin{abstract}
Factorial experiments are ubiquitous in the social and biomedical sciences, but when units fail to comply with each assigned factors, identification and estimation of the average treatment effects become impossible without strong assumptions.
Leveraging an instrumental variables approach, previous studies have shown how to identify and estimate the causal effect of treatment uptake among respondents who comply with treatment. 
A major caveat is that these identification results rely on strong assumptions on the effect of randomization on treatment uptake.
This paper shows how to bound these complier average treatment effects for bounded outcomes under more mild assumptions on non-compliance.  
\end{abstract}

\baselineskip=1.65\baselineskip
\clearpage

\section{Introduction}
 
Factorial experiments are widespread in the social and biomedical sciences for assessing how different interventions (or different features of an intervention) affect a unit's response. Unfortunately, in many of these experiments, units fail to comply with one or more of the factors, making it difficult to identify and estimate the actual effect of treatment uptake rather than the effect of treatment randomization. Instrumental variables (IV) has been a common approach to estimating causal effects in randomized experiments with non-compliance \citep{AngImbRub96}. The key intuition of this approach is that while the average effect of treatment uptake is unidentified, with some additional assumptions we may be able to identify those effects among \emph{compliers}, or those who would always comply with the treatment assignment they receive.

With a binary treatment, these additional assumptions, often called the instrumental variables (IV) assumptions, can be benign or be ensured to hold through the design of the experiment.
In factorial settings, however, point identification of complier average effects has required a treatment exclusion restriction that restricts the assignment of each factor to only affect that factor \citep{Blackwell17, blackwell2021noncompliance}.
This assumption can be strong as it rules out plausible interactions or spillovers between factors in the first stage.
Further,  \cite{Blackwell17} and \cite{blackwell2021noncompliance} focused on these strong assumptions holding for all factors.
However, in practice researchers often have a focal factor of interest and may only believe or want to assume that behavior with respect to that factor adheres to the assumptions, indicating a more nuanced approach may be appropriate.
In this paper, we derive partial identification bounds on complier average factorial effects under weaker assumptions about non-compliance which are factor specific.
We focus on two assumptions: for the factor of interest, we assume (1)  the existence of a ``least compliant profile'' of other factors; and (2) a weak exclusion restriction that allows for spillovers between factors on treatment uptake.
We use these assumption to bound the main effects of factors among units that always comply with assignment on that factor, whom we call the \emph{constant compliers}. 
These bounds can be informative of the direction of the constant complier main effects.  
Several studies have investigated partial identification bounds in the IV context. \cite{BalPea97} used linear programming to obtain bounds on the overall average treatment effects under the IV assumptions with a binary treatment. \cite{cheng2006bounds} derived bounds on the complier average treatment effect in a three-arm trial under an additional ``compliance monotonicity'' assumption that we generalize to the ``least compliant profile'' assumption for the factorial setting. \cite{song2024instrumental} derive bounds on the joint counterfactual distribution of the outcome when the treatment and outcome are discrete with a finite number of states. 
Our analysis leverages the factorial structure to generate plausible assumptions and valid bounds.  

\section{Setting}
Consider an experiment on $n$ units with $K$ two-level factors, with levels labeled as $-1$ and $+1$, and the potential for noncompliance on any of the factors.
There are $J = 2^K$ possible treatment assignments, the number of combinations of the $K$ factors, which we label $\Z = \{\z_1, \dots, \z_J\}$ where each $\z = (z_1,\dots,z_K) \in \{-1, 1\}^K$.
Let $\bm{Z}_i$ be the treatment assigned to unit $i$.
Further, let $\z_{-k}\in \{-1, 1\}^{K-1}$ denote a treatment assignment less the $k$th factor and similarly $\Z_{-k}$ be the set of possible assignments on $K-1$ factors.
\begin{assumption}[Complete randomization]\label{assump:cr}
We assume assignment to the treatment combinations is completely at random, fixing the number of units assigned to treatment $\z$ as $N_{\z}$ such that $N_{\z} \geq 2$ for all $\z \in \Z$ and $\sum_{\z \in \Z} N_{\z} = N$.
\end{assumption}

See \cite{DasPilRub15} for more on standard factorial designs.
Because of the noncompliance, we define $\bm{D}_i(\z) = (D_{i1}(\z), D_{i2}(\z), \dots, D_{iK}(\z))$ to be the treatment uptake for  unit $i$ on each factor 1 through $K$ if assigned to treatment $\z$ where each $D_{ik}(\z) \in \{-1,1\}$.

Assuming the Stable Unit Treatment Value Assumption (SUTVA) that there are no hidden forms of treatment and no interference between units, each unit has a potential outcome for each combination of assignment and uptake, $Y_i(\z, \bm{d})$.
We assume that $Y_i(\z, \bm{d}) \in [0,1]$, although extensions of our results to other bounded outcomes is straightforward.

\begin{assumption}[Exclusion restriction]\label{assump:er}
For $\z, \z'$ such that $D_i(\z) = D_i(\z')$, $Y_i(\z, D_i(\z)) = Y_i(\z', D_i(\z'))$.
\end{assumption}

Under Assumption~\ref{assump:er}, each unit has only $2^K$ potential outcomes, one for each possible treatment assignment.
Define the average of potential outcomes under treatment $\z$ as $\overline{Y}(\z) = N^{-1}\sum_{i=1}^NY_i(\z)$ and the average uptake value for factor $k$ under treatment $\z$ as $\overline{D}_k(\z) = N^{-1}\sum_{i=1}^ND_{ik}(\z)$.

Throughout, we take there to be a focal factor $k$ of interest that researchers wish to learn a complier average effect for, as opposed to an intent to treat effect.
Identification for interactions and implications for conducting inference for complier average effects on multiple factors are discussed in Supplementary Materials~\ref{sec:int}.

\subsection{Compliance types in factorial experiments}

A compliance type is a complete description of a unit's uptake for every possible treatment assignment, or the set $\{\bm{D}_{i}(\z) \mid \P(\bm{Z}_{i} = \z) > 0\}$. Without any restrictions on compliance, there are $(2^K)^{2^K}$ such types, typically making it difficult to learn anything about the effect of treatment uptake. A common solution to this problem in the binary IV literature is to restrict the compliance types in such a way that we can identify the effect of treatment uptake for some subset of compliance types. One common assumption is monotonicity, which assumes that the effect of treatment assignment on uptake is never negative. \cite{blackwell2021noncompliance} generalized this assumption to the factorial setting by having this hold conditionally for each factor. We write $(\z_{-k}, z_{k}^{+})$ and $(\z_{-k}, z_{k}^{-})$ as the $\z$ vector with the $k$th entry as $+1$ and $-1$, respectively.
\begin{assumption}[Conditional monotonicity for factor $k$]\label{assump:mono}
For each $\z_{-k} \in \mathcal{Z}_{-k}$, $D_{ik}(\z_{-k}, z_{k}^{+}) \geq D_{ik}(\z_{-k}, z_{k}^{-}) $ for all $i$.
\end{assumption}

 A unit's compliance might vary as a function of assignment to the other factors if those other factors can change compliance behavior.
With conditional monotonicity, we can characterize the compliance types for a particular factor $k$ as they vary by assignment on the other factors, $\z_{-k}$, as follows:
Compliers (c) for $\z_{-k}$ have $D_{ik}(\z_{-k}, z_{k}^{+}) = +1$ and $D_{ik}(\z_{-k}, z_{k}^{-}) = -1$; always-takers (a) have $D_{ik}(\z_{-k}, z_{k}^{+}) = D_{ik}(\z_{-k}, z_{k}^{-}) = +1$; and never takers have $D_{ik}(\z_{-k}, z_{k}^{+}) = D_{ik}(\z_{-k}, z_{k}^{-}) = -1$. 
For example, in a $2^2$ experiment, there are nine compliance types for factor 1 denoted by their type when assigned to $+1$ for factor 2 and $-1$ for factor 2: \{cc, ca, cn, ac, aa, an, nc, na, nn\}.

We focus our attention on units that would comply with treatment assignment on factor $k$ no matter the treatment assignment on other factors (those of type cc in the prior example).
We call these units the \emph{constant compliers} for factor $k$ and define them as the units with
$$
D_{ik}(\z_{-k}, z_{k}^{+}) - D_{ik}(\z_{-k}, z_{k}^{-}) = 2, \quad\forall \; \z_{-k} \in \mathcal{Z}_{-k}.
$$
Let $C_{ik} = 1$ if unit $i$ is a constant complier for factor $k$.
Constant compliers are equivalent to marginal compliers in \cite{blackwell2021noncompliance}.
However, in our setting there are individuals who comply with factor $k$ under some $\z_{-k}$ but not others, as discussed below, which is not allowed under the assumptions in \cite{blackwell2021noncompliance}.
Thus, we use the term constant complier to emphasize that these are individuals who always comply with factor $k$.
Define $\rho_{c_{k}} = \frac{1}{N} \sum_{i=1}^{N} C_{{ik}}$ as the proportion of units that are constant compliers for factor $k$ and define the average potential outcome under treatment assignment $\z$ for the constant compliers for factor $k$ as
$$
\overline{Y}_{c_{k}}(\z) = \frac{1}{\rho_{c_{k}}N} \sum_{i=1}^{N} C_{ik}Y_{i}(\z).
$$

We will group other compliance groups together by whether they ever comply with assignment on factor $k$.
For any particular assignment, $\z_{-k}$, we can partition units into constant compliers, those who comply for $\z_{-k}$ but fail to comply for at least one other assignment, and those who do not comply for $\z_{-k}$.
We call the latter two groups the \emph{conditional compliers} for $\z_{-k}$ and the \emph{conditional noncompliers} for $\z_{-k}$, and denote their shares in the population as $\rho_{cc(\z_{-k})}$ and $\rho_{cn(\z_{-k})}$, respectively. 
The notation reflects that these groups are defined based on their behavior with respect to assignment $\z_{-k}$, whereas constant compliers are the group that complies regardless of $\z_{-k}$.
We denote their average potential outcomes as $\overline{Y}_{cc(\z_{-k})}(\z)$ and $\overline{Y}_{cn(\z_{-k})}(\z)$. 

Our identification results focus on the main effect of factor $k$ among constant compliers:
\[\delta_k = \frac{1}{2^{K-1}}\sum_{\bm{z}_{{-k}} \in \mathcal{Z}_{-k}}\left[\overline{Y}_{c_k}(\z_{-k}, z_{k}^{+}) - \overline{Y}_{c_k}(\z_{-k}, z_{k}^{-})\right].\]

\section{Identification of main effects for constant compliers}\label{sec:ident}

What we can learn about the average factorial effects for constant compliers on factor $k$ will depend on the assumptions about compliance that we can maintain. 
\cite{blackwell2021noncompliance} assumed the Treatment Exclusion Restriction, that compliance for one factor does not vary with the assignment of other factors (and assumed this holds for all factors), and obtained point identification. 
Given the strength of that assumption, here we consider weaker assumptions that might be applicable in a wide variety of empirical settings and only need hold for factor $k$.

We begin with a useful decomposition of the intent-to-treat (ITT) effect of factor $k$ conditional on assignment on the other factors, $\z_{-k}$. Recall that the ITT is the effect of assignment on the outcome, ignoring uptake entirely, which we can decompose as
\begin{equation}\label{eq:itt_decomp}
  \begin{aligned}
    \gamma_k(\z_{-k}) &\equiv \overline{Y}(\z_{-k}, z_{k}^{+}) - \overline{Y}(\z_{-k}, z_{k}^{-})
                      = \rho_{c_{k}}\left[\overline{Y}_{c_k}(\z_{-k}, z_{k}^{+}) - \overline{Y}_{c_k}(\z_{-k}, z_{k}^{-})\right] \\ &+ \rho_{cc(\z_{-k})}\left[\overline{Y}_{cc(\z_{-k})}(\z_{-k}, z_{k}^{+}) - \overline{Y}_{cc(\z_{-k})}(\z_{-k}, z_{k}^{-})\right] \\
    &  + \rho_{cn(\z_{-k})}\left[\overline{Y}_{cn(\z_{-k})}(\z_{-k}, z_{k}^{+}) - \overline{Y}_{cn(\z_{-k})}(\z_{-k}, z_{k}^{-})\right].
  \end{aligned}
\end{equation}
We can see that the ITT is a mixture of three effects: (1) the effect for constant compliers, (2) the effect for conditional compliers, and (3) the effect for conditional noncompliers. 
It may seem odd to have the effect for conditional noncompliers, but recall that without further assumptions $Z_{ik}$ may affect compliance on other factors ($D_{ij}$ for $j \neq k$) even when it has no effect on $D_{ik}$.

To bring a more compact notation to this framework, we define a series of vectors that encode which treatments are contrasted for each factorial effect \citep{DasPilRub15, blackwell2021noncompliance}.
We begin by defining $\bm{Y} = (\overline{Y}(\bm{z}_1), \dots, \overline{Y}(\bm{z}_J))^T$, which is the vector of average potential outcomes for different treatment assignments. We then defined a series of $J$-length vectors, $\bm{g}_{j}$, each of which has one half of its entries as $+1$ and the other half as $-1$ to select the appropriate entries of $\bm{Y}$ when we use the product $\bm{g}^{T}_{j}\bm{Y}$.
The vectors $\bm{g}_{1}, \ldots, \bm{g}_{K}$ correspond to the main factorial effects of the $K$ factors.
For example, $\bm{g}_{k}$, for $k \leq K$, encodes the $k$th main effect so any treatment assignment vector $(\z_{-k}, z_{k}^{+})$ has a corresponding $+1$ entry in $\bm{g}_{k}$ and any assignment vector $(\z_{-k}, z_{k}^{-})$ has a corresponding $-1$ entry in $\bm{g}_{k}$.

With this notation in hand, we can write the main factorial ITT for factor $k$ as
$$
\frac{1}{2^{K-1}} \sum_{\z_{-k} \in \mathcal{Z}_{-k}} \gamma_{k}(\z_{-k}) = \frac{1}{2^{K-1}} \bm{g}_{k}^{T}\bm{Y}.  
$$
Plugging this into Equation~(\ref{eq:itt_decomp}), we can write the constant complier main effect as
$$
\begin{aligned}
\delta_{k} = \frac{1}{2^{K-1}\rho_{c_{k}}} \Bigg[ &\bm{g}_{k}^{T}\bm{Y} - \sum_{\z_{-k} \in \mathcal{Z}_{-k}} \rho_{cc(\z_{-k})}\left[\overline{Y}_{cc(\z_{-k})}(\z_{-k}, z_{k}^{+}) - \overline{Y}_{cc(\z_{-k})}(\z_{-k}, z_{k}^{-})\right]
 \\
& - \sum_{\z_{-k} \in \mathcal{Z}_{-k}} \rho_{cn(\z_{-k})}\left[\overline{Y}_{cn(\z_{-k})}(\z_{-k}, z_{k}^{+}) - \overline{Y}_{cn(\z_{-k})}(\z_{-k}, z_{k}^{-})\right] \Bigg].  
\end{aligned}
$$

Consider the ITT effect on treatment uptake.
Let $\nu_k^{+}(\z_{-k}) = \frac{1}{2}(\overline{D}_{k}(\z_{-k}, z_{k}^{+}) + 1)$ be the proportion of units receiving $+1$ for factor $k$ under assignment $(\z_{-k}, z_{k}^{+})$ and $\nu_k^{-}(\z_{-k}) = \frac{1}{2}(\overline{D}_{k}(\z_{-k}, z_{k}^{-}) + 1)$ be the proportion of units receiving $+1$ for factor $k$ under assignment $(\z_{-k}, z_{k}^{-})$.  Then we can write the ITT for treatment uptake as
\[\nu_k(\z_{-k}) = \nu_k^{+}(\z_{-k}) - \nu_k^{-}(\z_{-k}) =\frac{1}{2}\left[\overline{D}_k(\z_{-k}, z_{k}^{+}) - \overline{D}_k(\z_{-k}, z_{k}^{-})\right] = \rho_k+ \rho_{cc(\z_{-k})},\]
which tells us the proportion of respondents that are either conditional or constant compliers. Without further assumptions, we cannot break apart these two groups, so we introduce an assumption on the compliance behavior.

\subsection{Bounds under a least compliant profile}

Without further assumptions, the bounds for the constant complier effects will be from $-1$ to $1$ because there is insufficient information about the size of the constant complier group.
To put structure on this problem, we first assume there is a profile of the other factors that is the worst in terms of compliance for factor $k$ for all units.

\begin{assumption}[Least compliant profile for factor $k$]\label{assump:order_comp}
For all $i$, there is a $\tilde{\z}_{-k} \in \mathcal{Z}_{-k}$ such that $D_{ik}(\z_{-k}, z_{k}^{+}) - D_{ik}(\z_{-k}, z_{k}^{-}) \geq D_{ik}(\tilde{\z}_{-k}, z_{k}^{+}) - D_{ik}(\tilde{\z}_{-k}, z_{k}^{-}) $ for all other $\z_{-k} \in \mathcal{Z}_{-k}$.
\end{assumption}
Assumption~\ref{assump:order_comp} implies that there is a ``most difficult'' profile of the other factors, $\tilde{\z}_{-k}$, in terms of compliance.
In other words, if a unit complies with assignment on factor $k$ under $Z_{i,-k} = \tilde{\z}_{-k}$, then we know it would comply with all other possible assignments.
This implies that we can identify the proportion of constant compliers as
$
\nu_k(\tilde{\z}_{-k}) =  \frac{1}{2}\left[\overline{D}_k(\tilde{\z}_{-k}, z_{k}^{+}) - \overline{D}_k(\tilde{\z}_{-k}, z_{k}^{-})\right] = \rho_{k}.
$
Then we have $\rho_{cc(\z_{-k})} = \nu_k(\z_{-k}) - \nu_k(\tilde{\z}_{-k})$ and $\rho_{cn(\z_{-k})} = 1 - \nu_k(\z_{-k})$.
Below, we discuss this and other assumptions in the context of our applied setting. 

Our bounding strategy will rely on average effects for conditional compliers being bounded between $-1$ and $1$. Furthermore, noncompliers actually provide additional information because we can observe certain average outcomes for those groups. For instance, suppose that a unit is an always-taker for factor $k$ with assignment $\z_{-k}$, so that $D_{ik}(\z_{-k}, z_{k}^{+}) = D_{ik}(\z_{-k}, z_{k}^{-}) = +1$.
Under conditional monotonicity, we can distinguish one such set of always-takers: those that take treatment when assigned control, $D_{i}(\z_{-k}, z_{k}^{-}) = +1$; similarly, conditional never-takers will be those with $D_{i}(\z_{-k}, z_{k}^{+}) = -1$.
Let $\overline{Y}_{a(\z_{-k})}(\z_{-k}, z_{k}^{-})$ and $\overline{Y}_{n(\z_{-k})}(\z_{-k}, z_{k}^{+})$ be the average potential outcomes for conditional always- and never-takers, respectively, which we can identify under conditional monotonicity.

\begin{proposition}\label{prop:order_comp_bounds}
    Suppose Assumptions~\ref{assump:cr} and \ref{assump:er} hold and Assumptions~\ref{assump:mono} and \ref{assump:order_comp} hold with respect to factor $k$. Then, $\delta_{k} \in [\frac{\widetilde{\delta}_{k}-b^{-}_k}{2^{{K-1}}\nu_k(\tilde{\z}_{-k})},\frac{\widetilde{\delta}_{k} + b^{+}_k}{2^{{K-1}}\nu_k(\tilde{\z}_{-k})}] \cap [-1, 1]$, where
  \begin{equation}
    \label{eq:order_comp_bounds}
    \begin{aligned}
    b^{-}_k & =  \Big[\sum_{\bm{z}_{-k}} \nu_{k}(\z_{-k}) - \nu_k(\tilde{\z}_{-k}) + \nu_{k}^{-}(\z_{-k}) \Big], \qquad
    b^{+}_k = \Big[\sum_{\bm{z}_{-k}} \nu_{k}(\z_{-k}) - \nu_k(\tilde{\z}_{-k}) + (1- \nu_{k}^{+}(\z_{-k})) \Big], \\
      \widetilde{\delta}_{k}&= \Big[\bm{g}_k^T\bm{Y}  - \sum_{\bm{z}_{-k}}(1- \nu_{k}^{+}(\z_{-k}))\overline{Y}_{n(\z_{-k})}(\z_{-k}, z_{k}^{+}) + \sum_{\bm{z}_{-k}}  \nu_{k}^{-}(\z_{-k})\overline{Y}_{a(\z_{-k})}(\z_{-k}, z_{k}^{-})\Big]. 
    \end{aligned}
  \end{equation}
\end{proposition}

\begin{remark}\label{remark:larger_simple_bounds}
  The  width of these bounds (assuming they are away from bounds of the parameter space) can be written as $\rho_{k}^{-1}(1 - \rho_{k} + \frac{1}{2^{K-1}} \sum_{\bm{z}_{-k}} \rho_{cc(\bm{z}_{-k})})$.
  The numerator here combines all non-constant compliers ($1-\rho_{k}$) and the average size of the conditional complier group ($\frac{1}{2^{K-1}} \sum_{\bm{z}_{-k}} \rho_{cc(\bm{z}_{-k})}$), whereas the denominator is a size of the constant complier group. Thus, the bounds shrink as the size of the constant complier group grows.

\end{remark}

\subsection{Bounds under weak treatment exclusion}

We now derive bounds for the constant complier effects when we place restrictions on the first stage relationship between treatment assignment and treatment take-up. In particular, we investigate an assumption that a factor's assignment can only affect overall treatment uptake if it affects uptake on that particular factor.

\begin{assumption}[Weak treatment exclusion for factor $k$]\label{assump:weak_treat_ex}
  For any $\z_{-k} \in \mathcal{Z}_{-k}$, if  $D_{ik}(\z_{-k}, z_{k}^{+}) = D_{ik}(\z_{-k}, z_{k}^{-})$, then $\bm{D}_i(\z_{-k}, z_{k}^{+}) = \bm{D}_i(\z_{-k}, z_{k}^{-})$ for all $i$. 
\end{assumption}

Assumption~\ref{assump:weak_treat_ex} is inspired by a similar assumption in the context of noncompliance and interference in \cite{imai2018causal} and is a weaker version of the treatment exclusion assumption from \cite{blackwell2021noncompliance}, which held that $D_{ik}(z_{k}, \bm{z}_{-k}) = D_{ik}(z_{k}, \bm{z}'_{-k})$ for all $z_{k}, \bm{z}_{-k}$, $\bm{z}'_{-k}$, and $k$.  \cite{blackwell2021noncompliance} used this stronger assumption to point identify the effects of various complier average factorial effects. 
One downside to the stronger treatment exclusion assumption, however, is that it forbids any interactions between factors on the relationship between $\bm{Z}_{i}$ and $\bm{D}_{i}$. 
That is, it disallows assignment of treatment for any factor $k$ from influencing uptake of any other factor $j$.
Assumption~\ref{assump:weak_treat_ex} allows for a structured interaction, where assignment to factor $k$ can only affect uptake on other factors if unit $i$ is a complier for $k$.
It also only requires this assumption to hold on the particular factor $k$ of interest.
If all units are constant compliers on all factors except for factor $k$, then Assumption~\ref{assump:weak_treat_ex} is automatically satisfied.

\begin{figure}
  \centering
  \subfloat[]{\label{fig:weak_tr_dag_a}
  \begin{tikzpicture}[>=triangle 45]
    \node (d1) at (0,1) {$D_{1}$};
    \node (z1) at (-2, 1) {$Z_{1}$};
    \node (z2) at (-2, -1) {$Z_{2}$};
    \node (d2) at (0, -1)  {$D_{2}$};
    \node (u3) at (1, 0) {$U_{3}$};

    \node (y) at (3, 0) {$Y$};
    \draw[->, >=stealth, thick] (d1) -- (y);
    \draw[->, >=stealth, thick] (d1) -- (d2);
    \draw[->, >=stealth, thick] (z1) -- (d1);
    \draw[->, >=stealth, thick] (d2) -- (y);
    \draw[->, >=stealth, thick] (z2) -- (d2);
    \draw[->, >=stealth, thick, dashed] (u3) -- (d2);
    \draw[->, >=stealth, thick, dashed] (u3) -- (d1);
    \draw[->, >=stealth, thick, dashed] (u3) -- (y);
  \end{tikzpicture}
}
\subfloat[]{\label{fig:weak_tr_dag_b}
  \begin{tikzpicture}[>=triangle 45]
    \node (d1) at (0,1) {$D_{1}$};
    \node (z1) at (-2, 1) {$Z_{1}$};
    \node (z2) at (-2, -1) {$Z_{2}$};
    \node (d2) at (0, -1)  {$D_{2}$};
    \node (u3) at (1, 0) {$U_{3}$};

    \node (y) at (3, 0) {$Y$};
    \draw[->, >=stealth, thick] (d1) -- (y);
    \draw[->, >=stealth, thick] (z1) -- (d1);
    \draw[->, >=stealth, thick] (z1) -- (d2);
    \draw[->, >=stealth, thick] (d1) -- (d2);
    \draw[->, >=stealth, thick] (d2) -- (y);
    \draw[->, >=stealth, thick] (z2) -- (d2);
    \draw[->, >=stealth, thick, dashed] (u3) -- (d2);
    \draw[->, >=stealth, thick, dashed] (u3) -- (d1);
    \draw[->, >=stealth, thick, dashed] (u3) -- (y);
  \end{tikzpicture}
  }
  \caption{\label{fig:weak_tr_dag} Directed acyclic graphs (DAGs) of causal structures that would satisfy weak treatment exclusion for both factors (a) and for just factor 2 (b).}
\end{figure}
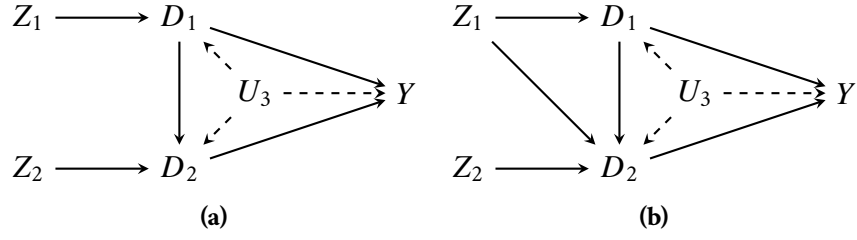

There are many situations where Assumption~\ref{assump:weak_treat_ex} can hold by design. In particular, Figure~\ref{fig:weak_tr_dag} shows two directed acyclic graphs (DAGs) with causal structures where the first factor is assigned first, allowing for it to affect the second factor but not vice versa.
In both DAGs, assignment to the second factor cannot affect uptake on the first factor, so weak exclusion for factor 2 holds. In Figure~\ref{fig:weak_tr_dag_a}, weak exclusion also holds with respect to factor 1 because assignment to factor 1 only affects uptake on factor 2 through uptake on factor 1. Thus, there is no effect of $Z_{i1}$ on $D_{i2}$ for noncompliers on factor 1. This DAG may be plausible in situations where noncompliers do not know their assignment, making it impossible for it to affect their compliance on factor 2. In Figure~\ref{fig:weak_tr_dag_b}, the direct arrow from $Z_{i1}$ to $D_{i2}$ means weak exclusion is violated for factor 1, but the assumption does continue to hold for factor 2.
In these settings we still allow for arbitrary joint dependence between treatment uptake on each factor and the outcome. The important part of Assumption~\ref{assump:weak_treat_ex} is the restrictions on the effects of assignment on uptake.  

\begin{remark}
For a proposed least compliant profile, Assumption~\ref{assump:order_comp} is falsifiable by testing whether the compliance rate under that profile is less than or equal to the compliance rate under any other profile.
As stated, we do not believe that Assumption~\ref{assump:weak_treat_ex} is directly testable, though stronger versions of the Assumption, such as those encoded in the DAGs of Figure~\ref{fig:weak_tr_dag_a}, could be tested by assessing the missing arrows from treatment assignment to uptake.
\end{remark}

Assumptions~\ref{assump:er} and \ref{assump:weak_treat_ex} give us that if unit $i$ does not comply on factor $k$ given assignment $\z_{-k}$ on the other factors, then $Y_{i}(\z_{-k}, z_{k}^{+}) = Y_{i}(\z_{k}, z_{k}^{-})$, which implies that there is no effect for conditional noncompliers and allows us to obtain narrower bounds. 

\begin{proposition}\label{prop:weak_tr_ex}
  Suppose Assumptions~\ref{assump:cr} and \ref{assump:er}  hold, and Assumptions~\ref{assump:mono}, \ref{assump:order_comp}, and~\ref{assump:weak_treat_ex} hold with respect to factor $k$.
  Then,  $\delta_{k} \in[\frac{1}{2^{K-1}}\sum_{\z_{-k}}\tilde{b}_{k}^-(\z_{-k}),\frac{1}{2^{K-1}}\sum_{\z_{-k}}\tilde{b}_{k}^+(\z_{-k})]$, where
       \begin{align*} 
\tilde{b}_{k}^-(\z_{-k}) &= \max\left(0, \frac{\nu_k^{+}(\z_{-k})}{\nu_k(\tilde{\z}_{-k})}\overline{Y}(\z_{-k}, z_{k}^{+}|D_k(\z_{-k}, z_{k}^{+}) = +1) -\frac{\nu_k^{-}(\z_{-k})}{\nu_k(\tilde{\z}_{-k})}\overline{Y}_{a(\z_{-k})}(\z_{-k}, z_{k}^{-}) -  \frac{\nu_k(\z_{-k}) -\nu_k(\tilde{\z}_{-k}) }{\nu_k(\tilde{\z}_{-k})}\right)\\
&  -  \min\left(1, \frac{ 1-\nu_k^{-}(\z_{-k})}{ \nu_k(\tilde{\z}_{-k})} \overline{Y}(\z_{-k}, z_{k}^{-}|D_k(\z_{-k}, z_{k}^{-}) = -1) - \frac{1 - \nu_k^{+}(\z_{-k})}{ \nu_k(\tilde{\z}_{-k})}\overline{Y}_{n(\z_{-k})}(\z_{-k}, z_{k}^{+})\right),
\end{align*}
\begin{align*}
&\tilde{b}_{k}^+(\z_{-k}) = \min\left(1,  \frac{\nu_k^{+}(\z_{-k})}{\nu_k(\tilde{\z}_{-k})}\overline{Y}(\z_{-k}, z_{k}^{+}|D_k(\z_{-k}, z_{k}^{+}) = +1) -\frac{\nu_k^{-}(\z_{-k})}{\nu_k(\tilde{\z}_{-k})}\overline{Y}_{a(\z_{-k})}(\z_{-k}, z_{k}^{-})\right)\\
& - \max\Bigg(0, \frac{ 1-\nu_k^{-}(\z_{-k})}{\nu_k(\tilde{\z}_{-k})} \overline{Y}(\z_{-k}, z_{k}^{-}|D_k(\z_{-k}, z_{k}^{-}) = -1)- \frac{1 - \nu_k^{+}(\z_{-k})}{\nu_k(\tilde{\z}_{-k})}\overline{Y}_{n(\z_{-k})}(\z_{-k}, z_{k}^{+}) - \frac{\nu_k(\z_{-k}) -\nu_k(\tilde{\z}_{-k}) }{\nu_k(\tilde{\z}_{-k})}\Bigg),
 \end{align*} 
 where $\bar{Y}(\z_{-k}, z_k|D_k(\z_{-k}, z_k) = w)$ is the average outcome among individuals assigned to $(\z_{-k}, z_k)$ whose uptake on factor $k$ is $w$.
 These bounds are narrower than under Proposition~\ref{prop:order_comp_bounds}. 
\end{proposition}
 
 \begin{remark}
 Under the same assumptions as in Proposition~\ref{prop:weak_tr_ex}, there are less sharp but more interpretable bounds: $\delta_{k} \in[\delta^{\prime}_k-\tilde{b}_k,\delta^{\prime}_k + \tilde{b}_k] \cap [-1, 1]$, where
  \begin{equation}
    \label{eq:weak_tr_ex_bounds}
    \begin{aligned}      
     \tilde{b}_k & =  \frac{\sum_{\z_{-k}} (\nu_k(\z_{-k}) - \nu_k(\tilde{\z}_{-k}))}{2^{K-1}\nu_k(\tilde{\z}_{-k})}, &\qquad 
    \delta^{\prime}_k & =  \frac{\bm{g}_{k}^{T}\bm{Y}}{2^{K-1}\nu_k(\tilde{\z}_{-k})} .
    \end{aligned} 
  \end{equation}
  \end{remark}

The width of the bounds in Proposition~\ref{prop:weak_tr_ex} depend on the size of the conditional compliers relative to the size of the constant compliers.
If the constant compliers group is too small (i.e., if we have weak instrument) or the conditional complier group too big, the resulting bounds may be uninformative.
Under the stronger assumption of treatment exclusion from \cite{blackwell2021noncompliance}, the quantities $\nu_k(\z_{-k}) - \nu_k(\tilde{\z}_{-k})$ are zero because there are only constant compliers and noncompliers, implying the bounds collapse to a point as in \cite{blackwell2021noncompliance}.
In Supplemental Materials~\ref{append:simulation}, we present a simulation study of these bounds that show how the width of the bounds increases with the number of factors and decreases with the rate of constant compliance; see also Remark~\ref{remark:increasing_k}.
One advantage of considering constant compliers for one factor $k$ over the perfect compliers defined in \cite{blackwell2021noncompliance}, who comply with all factors, is that the constant complier group will always be at least as large (and typically larger) than the perfect complier group.
Without the stronger assumptions of \cite{blackwell2021noncompliance} or very high compliance rates across all factors, it may be difficult to get informative bounds for multiple factors in practice.
A disadvantage of constant compliers over perfect compliers is that estimands for different factors target different populations; our work is therefore most appropriate when researchers wish to learn a complier effect for a focal factor of interest.

In Supplementary Materials~\ref{append:sharp_bounds}, we use linear programming to obtain sharp bounds under a binary outcome with two factors.
Those bounds overlap with the bounds above for some parameter values.
The bounds above do not depend on binary outcomes and are, thus, more general.

\begin{remark}\label{remark:increasing_k}
An interesting nuance of the factorial setting is the consequence of increasing the number of factors, $K$.
There are two potential consequences for the identification results.
First, the assumptions, especially \ref{assump:order_comp} and \ref{assump:weak_treat_ex}, may be more difficult to justify, because increasing the number of factors could lead to more heterogeneity of behavior across the possible treatments and more complicated relationships among the factors.
Second, as additional factors are added to an experiment, we might expect the proportion of constant compliers for a given factor to decrease.
That is, we expect the proportion of individuals who always comply with factor $k$ when varying the levels of $K$ factors to be less than or equal to the proportion of individuals who always comply with factor $k$ when varying the levels of $K-1$ factors.
If this is the case and the average proportion of conditional compliers remains the same or increases, the bounds will increase (or at least not decrease) in width as $K$ increases and $\nu_k(\tilde{\z}_{-k})$ decreases (or remains the same).
\end{remark}
\begin{remark}\label{remark:conservative}
Consider the case where we believe Assumption~\ref{assump:order_comp} holds but we do not know which combination of other factors, $\tilde{\z}_{-k}$ leads to the ``least'' compliance.
A seemingly reasonable strategy might be to estimate $\rho_k$ based on the assignment of other factors with the smallest \textit{observed} compliance.
That is, $\widehat{\rho}_k = \min_{\z_{-k}} \widehat{\overline{D}}_k(\z_{-k}, z_{k}^{+}) -\widehat{ \overline{D}}_k(\z_{-k}, z_{k}^{-})$.
This is not unbiased for $\rho_k$, but rather will equal in expectation some value $t_k \leq \rho_k$. 
The question is: Will this be conservative, in the sense that the bound resulting from using $t_k$ will cover the bound using $\rho_k$?
As we show in the Supplemental Materials~\ref{append:proofs}, it is, in fact, conservative (and will be wider than using the true value $\rho_k$).
Two further comments on this results are warranted: 
1. Although the bounds using $t_k$ are conservative in the way described, they will be centered on an estimate strictly greater in magnitude due to dividing by a smaller proportion.
2. This conservativeness only holds if Assumption~\ref{assump:order_comp} holds. Importantly, this assumption is for every unit in the experiment and thus not directly testable. Arguably, if researchers do not have a substantive reason to believe a particular treatment combination $\z_{-k}$ will yield lower compliance on factor $k$ for all units (or if the observed data provides strong evidence against the hypothesized combination), then the assumption in general might be called into question. 
We investigate this approach in our simulation study in Section~D of the Supplemental Materials and find that bounds using the estimated least compliant profile are indeed more conservative than with the known least compliant profile. In that data generating process, the conservative bounds are roughly 20\% wider than the true bounds with the known least compliant profile.
\end{remark}

\section{Data illustration}\label{sec:data_ex}
We reanalyze data from \cite{BlaJamShe17}, an experiment on the use of cognitive behavioral therapy (CBT) and cash infusions for antisocial and criminal behaviors among at-risk young men in Liberia, to illustrate our methods.
We focus on the outcome of reported drug selling behavior following the experimental intervention assigned, which was a binary outcome (yes or no). 
The experiment was a $2^2$ factorial design, with factors CBT and cash each having levels receive or do not receive.
Compliance was one-sided in that individuals assigned not to receive an intervention (CBT or cash) would not be able to receive that intervention.
There was some noncompliance on both factors: 63\% overall compliance on CBT and 98\% overall compliance on cash.

Assumption~\ref{assump:mono} is assured for both factors due to the one-sided nature of the compliance.
Moreover, because of one-sided noncompliance, there are 4 compliance types for CBT (or cash): \{cc, cn, nc, nn\}, where the first type refers to CBT (or cash) compliance when assigned to receive cash (or CBT) (complier is c, never taker is n) and the second type is CBT (or cash) compliance when not assigned to receive cash (or CBT).

Consider possible compliance behaviors and how they relate to the least compliant profile for factor $k$ assumption (Assumption~\ref{assump:order_comp}) and the weak-treatment exclusion restriction assumption (Assumption~\ref{assump:weak_treat_ex}).
It is plausible that individuals' compliance with cash is not impacted by their assignment on CBT.
In this case, Assumption~\ref{assump:order_comp} automatically holds for cash.
On the other hand, being assigned to receive a cash transfer might change compliance with CBT assignment. 
A unit might be a conditional never-taker for CBT if not assigned the cash transfer, but a conditional complier if they were assigned the cash transfer; this is plausible if individuals feel more obligation to comply after receiving cash.
If anyone that complies with CBT in the ``no cash'' condition will comply with CBT in the ``cash'' condition, then Assumption~\ref{assump:order_comp} holds for CBT with no cash transfer being the level giving least compliance.
This aligns with the observed data: 64.4\% of individuals assigned to CBT when also assigned to cash, whereas 62.4\% of those not assigned to cash complied with CBT.
Under the scenarios laid out, Assumption~\ref{assump:weak_treat_ex} would hold for CBT but not for cash. 
These assumptions must always be assessed based on subject-matter knowledge.
Note that the stronger treatment exclusion restriction assumption of \cite{Blackwell17} and \cite{blackwell2021noncompliance} would require no interactions of any kind between uptake on CBT and assignment on cash -- meaning the type of ``obligation'' behavior previously described would not be allowed.

Under Proposition~\ref{prop:weak_tr_ex}, the estimated bound for the constant complier effect on CBT is $(-0.114, -0.081)$ and that for the constant complier effect on cash is $(-0.019, -0.012)$.
The bounds on both effects do not include 0, making them informative.
Under the stronger assumptions of  \cite{blackwell2021noncompliance}, the estimated effect for CBT is $-0.096$ and for cash is $-0.015$, which both fall within the previous bounds.
The bound under Proposition~\ref{prop:order_comp_bounds} (that is, without assuming weak treatment exclusion) for CBT is $(-0.123, 0.4497)$ and under Remark~\ref{remark:larger_simple_bounds} is $(-0.701, 0.506)$---both much wider and including 0.
Under Assumption~\ref{assump:weak_treat_ex} for CBT but not cash, we would not trust the bounds (or the stronger point identification) on the cash estimate.
However, it provides a useful illustration here in comparison to CBT: as the compliance rate increases, the bounds and confidence intervals will both become narrower.
Further, the bounds depend on estimates of $\nu_k(\z_{-k}) - \nu_k(\tilde{\z}_{-k})$, so the width of the bounds will also depend on how much variability there is in observed compliance across the treatment arms.

\bibliographystyle{apalike}
\bibliography{ref}{}

\if0\submission
\begin{appendices}
\renewcommand{\thesection}{\Alph{section}}
\renewcommand\thefigure{SM.\arabic{figure}}
\renewcommand\thetable{SM.\arabic{table}}
\setcounter{lemma}{0}
\renewcommand{\thelemma}{SM.\arabic{lemma}}
\renewcommand{\theassumption}{SM.\arabic{assumption}}
\renewcommand{\theproposition}{SM.\arabic{proposition}}

These supplementary materials are organized as follows:
Section~\ref{append:proofs} provides proofs of the results in the main paper.
Section~\ref{append:int} gives additional results for bounds on two-way interaction effects.
Section~\ref{append:est_inf} discusses estimation and inference for the identification methods discussed in the main paper.

\section{Proofs}\label{append:proofs}

\subsection{Proof of Proposition~\ref{prop:order_comp_bounds}}\label{append:proof_rder_comp_bounds}
We prove the bounds under Assumptions~\ref{assump:cr}, \ref{assump:er}, \ref{assump:mono}, and~\ref{assump:order_comp}.
\begin{proof}
Given the results shown in Section~\ref{sec:ident}, we can rewrite $\delta_k$ as 
   \begin{align*}
     \delta_{k} =& \frac{1}{2^{K-1}\nu_k(\tilde{\z}_{-k})} \Bigg[ \bm{g}_{k}^{T}\bm{Y} - \sum_{\z_{-k} \in \mathcal{Z}_{-k}} (\nu_k(\z_{-k}) - \nu_k(\tilde{\z}_{-k}))\left[\overline{Y}_{cc(\z_{-k})}(\z_{-k}, z_{k}^{+}) - \overline{Y}_{cc(\z_{-k})}(\z_{-k}, z_{k}^{-})\right]
  \\
 & \hspace{2.8cm} - \sum_{\z_{-k} \in \mathcal{Z}_{-k}} (1 - \nu_k(\z_{-k})) \left[\overline{Y}_{cn(\z_{-k})}(\z_{-k}, z_{k}^{+}) - \overline{Y}_{cn(\z_{-k})}(\z_{-k}, z_{k}^{-})\right] \Bigg]\\
 =& \frac{1}{2^{K-1}\nu_k(\tilde{\z}_{-k})} \Bigg[ \bm{g}_{k}^{T}\bm{Y} - \sum_{\z_{-k} \in \mathcal{Z}_{-k}} (\nu_k(\z_{-k}) - \nu_k(\tilde{\z}_{-k}))\left[\overline{Y}_{cc(\z_{-k})}(\z_{-k}, z_{k}^{+}) - \overline{Y}_{cc(\z_{-k})}(\z_{-k}, z_{k}^{-})\right]
  \\
 & \hspace{2.8cm} - \sum_{\z_{-k} \in \mathcal{Z}_{-k}} \rho_{n(\z_{-k})} \left[\overline{Y}_{n(\z_{-k})}(\z_{-k}, z_{k}^{+}) - \overline{Y}_{n(\z_{-k})}(\z_{-k}, z_{k}^{-})\right] \\
 & \hspace{2.8cm} - \sum_{\z_{-k} \in \mathcal{Z}_{-k}} \rho_{a(\z_{-k})} \left[\overline{Y}_{a(\z_{-k})}(\z_{-k}, z_{k}^{+}) - \overline{Y}_{a(\z_{-k})}(\z_{-k}, z_{k}^{-})\right]\Bigg].  
 \end{align*}

 Rearrange to put the identified components first as follows:
  \begin{align*}
     \delta_{k}  =& \frac{1}{2^{K-1}\nu_k(\tilde{\z}_{-k})} \Bigg[ \bm{g}_{k}^{T}\bm{Y} - \sum_{\z_{-k} \in \mathcal{Z}_{-k}}\left( \rho_{n(\z_{-k})} \overline{Y}_{n(\z_{-k})}(\z_{-k}, z_{k}^{+}) - \rho_{a(\z_{-k})}\overline{Y}_{a(\z_{-k})}(\z_{-k}, z_{k}^{-})\right)\\
     &  \hspace{2.8cm} - \sum_{\z_{-k} \in \mathcal{Z}_{-k}} (\nu_k(\z_{-k}) - \nu_k(\tilde{\z}_{-k}))\left[\overline{Y}_{cc(\z_{-k})}(\z_{-k}, z_{k}^{+}) - \overline{Y}_{cc(\z_{-k})}(\z_{-k}, z_{k}^{-})\right]
  \\
 & \hspace{2.8cm} - \sum_{\z_{-k} \in \mathcal{Z}_{-k}}\left[\rho_{a(\z_{-k})} \overline{Y}_{a(\z_{-k})}(\z_{-k}, z_{k}^{+}) - \rho_{n(\z_{-k})} \overline{Y}_{n(\z_{-k})}(\z_{-k}, z_{k}^{-})\right] \Bigg].  
 \end{align*}
 
 We can bound the unidentified  components,
 \begin{align*}
      &  - \sum_{\z_{-k} \in \mathcal{Z}_{-k}} (\nu_k(\z_{-k}) - \nu_k(\tilde{\z}_{-k}))\left[\overline{Y}_{cc(\z_{-k})}(\z_{-k}, z_{k}^{+}) - \overline{Y}_{cc(\z_{-k})}(\z_{-k}, z_{k}^{-})\right]
  \\
 & \qquad - \sum_{\z_{-k} \in \mathcal{Z}_{-k}}\left[\rho_{a(\z_{-k})} \overline{Y}_{a(\z_{-k})}(\z_{-k}, z_{k}^{+}) - \rho_{n(\z_{-k})} \overline{Y}_{n(\z_{-k})}(\z_{-k}, z_{k}^{-})\right] \\
 & \in \Bigg(  - \sum_{\z_{-k} \in \mathcal{Z}_{-k}} (\nu_k(\z_{-k}) - \nu_k(\tilde{\z}_{-k})) - \sum_{\z_{-k} \in \mathcal{Z}_{-k}}\rho_{a(\z_{-k})}, \\
 & \qquad\sum_{\z_{-k} \in \mathcal{Z}_{-k}} (\nu_k(\z_{-k}) - \nu_k(\tilde{\z}_{-k})) + \sum_{\z_{-k} \in \mathcal{Z}_{-k}}\rho_{n(\z_{-k})} \Bigg).
 \end{align*}
We can finish the proof by noting that $\rho_{a(\z_{-k})} = \nu_{k}^{-}(\z_{-k})$ and $\rho_{n(\z_{-k})} = 1 - \nu_{k}^{+}(\z_{-k})$. 
 
\end{proof}



\subsection{Proof of Proposition~\ref{prop:weak_tr_ex}}\label{append:proof_weak_tr_ex}
We prove the bounds under Assumptions~\ref{assump:cr}, \ref{assump:er}, \ref{assump:mono}, \ref{assump:order_comp}, and~\ref{assump:weak_treat_ex}.
\begin{proof}
Much of the proof is similar to the approach in, e.g., \cite{cheng2006bounds}.
In the formulas below, we will condition on $c_k$, $cc$, $a$, and $n$ to indicate conditioning on the respective compliance type given $\z_{-k}$.
We have the following decompositions for identifiably means:
\begin{align*}
\overline{Y}(\z_{-k}, z_{k}^{+}|D_k(\z_{-k}, z_{k}^{+}) = +1) &= \frac{\rho_{c_{k}}}{\rho_{c_{k}} + \rho_{cc(\z_{-k})} + \rho_{a(\z_{-k})}}\overline{Y}_{c_k}(\z_{-k}, z_{k}^{-})\\
& \qquad + \frac{\rho_{cc(\z_{-k})} }{\rho_{c_{k}} + \rho_{cc(\z_{-k})} + \rho_{a(\z_{-k})}}\overline{Y}_{cc(\z_{-k})}(\z_{-k}, z_{k}^{+}) \\
& \qquad+ \frac{\rho_{a(\z_{-k})} }{\rho_{c_{k}} + \rho_{cc(\z_{-k})} + \rho_{a(\z_{-k})}}\overline{Y}_{a(\z_{-k})}(\z_{-k}, z_{k}) \text{ for } z_k \in \{-1,+1\},\\
\overline{Y}(\z_{-k}, z_{k}^{+}|D_k(\z_{-k}, z_{k}^{+}) = -1) &= \overline{Y}_{n(\z_{-k})}(\z_{-k}, z_{k}) \text{ for } z_k \in \{-1,+1\},\\
\overline{Y}(\z_{-k}, z_{k}^{-}|D_k(\z_{-k}, z_{k}^{-}) = +1) &=\overline{Y}_{a(\z_{-k})}(\z_{-k}, z_{k}) \text{ for } z_k \in \{-1,+1\},\\
\overline{Y}(\z_{-k}, z_{k}^{-}|D_k(\z_{-k}, z_{k}^{-}) = -1) &= \frac{\rho_{c_{k}}}{\rho_{c_{k}} + \rho_{cc(\z_{-k})} + \rho_{n(\z_{-k})}}\overline{Y}_{c_k}(\z_{-k}, z_{k}^{-})\\
& \qquad + \frac{\rho_{cc(\z_{-k})} }{\rho_{c_{k}} + \rho_{cc(\z_{-k})} + \rho_{n(\z_{-k})}}\overline{Y}_{cc(\z_{-k})}(\z_{-k}, z_{k}^{-}) \\
& \qquad+ \frac{\rho_{n(\z_{-k})} }{\rho_{c_{k}} + \rho_{cc(\z_{-k})} + \rho_{n(\z_{-k})}}\overline{Y}_{n(\z_{-k})}(\z_{-k}, z_{k}) \text{ for } z_k \in \{-1,+1\}.
\end{align*}

We can simplify the first and last equation to 
\begin{align*}
&\overline{Y}(\z_{-k}, z_{k}^{+}|D_k(\z_{-k}, z_{k}^{+}) = +1) -  \frac{\rho_{a(\z_{-k})} }{\rho_{c_{k}} + \rho_{cc(\z_{-k})} + \rho_{a(\z_{-k})}}\overline{Y}(\z_{-k}, z_{k}^{-}|D_k(\z_{-k}, z_{k}^{-}) = +1)\\
&= \frac{\rho_{c_{k}}}{\rho_{c_{k}} + \rho_{cc(\z_{-k})} + \rho_{a(\z_{-k})}}\overline{Y}_{c_k}(\z_{-k}, z_{k}^{-}) + \frac{\rho_{cc(\z_{-k})} }{\rho_{c_{k}} + \rho_{cc(\z_{-k})} + \rho_{a(\z_{-k})}}\overline{Y}_{cc(\z_{-k})}(\z_{-k}, z_{k}^{+}) \\
& \implies \frac{\rho_{c_{k}} + \rho_{cc(\z_{-k})} + \rho_{a(\z_{-k})}}{\rho_{c_{k}} + \rho_{cc(\z_{-k})}}\overline{Y}(\z_{-k}, z_{k}^{+}|D_k(\z_{-k}, z_{k}^{+}) = +1)\\
&\qquad \qquad -  \frac{\rho_{a(\z_{-k})} }{\rho_{c_{k}} + \rho_{cc(\z_{-k})}}\overline{Y}(\z_{-k}, z_{k}^{-}|D_k(\z_{-k}, z_{k}^{-}) = +1)\\
&\qquad = \frac{\rho_{c_{k}}}{\rho_{c_{k}} + \rho_{cc(\z_{-k})}}\overline{Y}_{c_k}(\z_{-k}, z_{k}^{+}) + \frac{\rho_{cc(\z_{-k})} }{\rho_{c_{k}} + \rho_{cc(\z_{-k})} }\overline{Y}_{cc(\z_{-k})}(\z_{-k}, z_{k}^{+}) \\
&\text{and}\\
&\overline{Y}(\z_{-k}, z_{k}^{+}|D_k(\z_{-k}, z_{k}^{-}) = -1) - \frac{\rho_{n(\z_{-k})} }{\rho_{c_{k}} + \rho_{cc(\z_{-k})} + \rho_{n(\z_{-k})}}\overline{Y}(\z_{-k}, z_{k}^{+}|D_k(\z_{-k}, z_{k}^{+}) = -1)\\
&= \frac{\rho_{c_{k}}}{\rho_{c_{k}} + \rho_{cc(\z_{-k})} + \rho_{n(\z_{-k})}}\overline{Y}_{c_k}(\z_{-k}, z_{k}^{-}) + \frac{\rho_{cc(\z_{-k})} }{\rho_{c_{k}} + \rho_{cc(\z_{-k})} + \rho_{n(\z_{-k})}}\overline{Y}_{cc(\z_{-k})}(\z_{-k}, z_{k}^{-})\\
&\implies \frac{\rho_{c_{k}} + \rho_{cc(\z_{-k})} + \rho_{n(\z_{-k})}}{\rho_{c_{k}} + \rho_{cc(\z_{-k})} }\overline{Y}(\z_{-k}, z_{k}^{+}|D_k(\z_{-k}, z_{k}^{-}) = -1)\\
&\qquad \qquad - \frac{\rho_{n(\z_{-k})} }{\rho_{c_{k}} + \rho_{cc(\z_{-k})} + }\overline{Y}(\z_{-k}, z_{k}^{+}|D_k(\z_{-k}, z_{k}^{+}) = -1)\\
&\qquad = \frac{\rho_{c_{k}}}{\rho_{c_{k}} + \rho_{cc(\z_{-k})} }\overline{Y}_{c_k}(\z_{-k}, z_{k}^{-}) + \frac{\rho_{cc(\z_{-k})} }{\rho_{c_{k}} + \rho_{cc(\z_{-k})}}\overline{Y}_{cc(\z_{-k})}(\z_{-k}, z_{k}^{-}).
\end{align*}

Let
\begin{align*}
q_{+}(\z_{-k}) &= \frac{\rho_{c_{k}} + \rho_{cc(\z_{-k})} + \rho_{a(\z_{-k})}}{\rho_{c_{k}} + \rho_{cc(\z_{-k})}}\overline{Y}(\z_{-k}, z_{k}^{+}|D_k(\z_{-k}, z_{k}^{+}) = +1) \\
&\qquad-  \frac{\rho_{a(\z_{-k})} }{\rho_{c_{k}} + \rho_{cc(\z_{-k})}}\overline{Y}(\z_{-k}, z_{k}^{-}|D_k(\z_{-k}, z_{k}^{-}) = +1)\\
q_{-}(\z_{-k}) &=
 \frac{\rho_{c_{k}} + \rho_{cc(\z_{-k})} + \rho_{n(\z_{-k})}}{\rho_{c_{k}} + \rho_{cc(\z_{-k})}}\overline{Y}(\z_{-k}, z_{k}^{-}|D_k(\z_{-k}, z_{k}^{-}) = -1) \\
 &\qquad-  \frac{\rho_{n(\z_{-k})} }{\rho_{c_{k}} + \rho_{cc(\z_{-k})}}\overline{Y}(\z_{-k}, z_{k}^{+}|D_k(\z_{-k}, z_{k}^{+}) = -1)\\
\end{align*}

For binary outcomes, we have 
\begin{align*}
0 \leq \E[Y(\z_{-k}, z_{k})| D_k(\z_{-k}, z_{k}^\prime)], \overline{Y}_{c_k}(\z_{-k}, z_{k}^{-}), \overline{Y}_{c_k}(\z_{-k}, z_{k}^{-}), \overline{Y}_{cc(\z_{-k})}(\z_{-k}, z_{k}^{-}) \leq 1.
\end{align*}

Then using results from \cite{horowitz1995identification} (or linear programming), we have
\begin{align*}
\max\left(0, 1 - \frac{1-q_{+}(\z_{-k})}{ \frac{\rho_{c_{k}}}{\rho_{c_{k}} + \rho_{cc(\z_{-k})}}}\right)  \leq \overline{Y}_{c_k}(\z_{-k}, z_{k}^{+})  \leq \min\left(1, \frac{q_{+}(\z_{-k})}{ \frac{\rho_{c_{k}}}{\rho_{c_{k}} + \rho_{cc(\z_{-k})}}}\right),\\
\max\left(0, 1 - \frac{1-q_{-}(\z_{-k})}{  \frac{\rho_{c_{k}}}{\rho_{c_{k}} + \rho_{cc(\z_{-k})} }}\right)  \leq \overline{Y}_{c_k}(\z_{-k}, z_{k}^{-})  \leq \min\left(1, \frac{q_{-}(\z_{-k})}{ \frac{\rho_{c_{k}}}{\rho_{c_{k}} + \rho_{cc(\z_{-k})} }}\right).
\end{align*}

Then we must have
\begin{align*}
\delta_k &=  \overline{Y}_{c_k}(\z_{-k}, z_{k}^{+}) -  \overline{Y}_{c_k}(\z_{-k}, z_{k}^{-})\\
&\in \Bigg(\max\left(0, 1 - \frac{1-q_{+}(\z_{-k})}{ \frac{\rho_{c_{k}}}{\rho_{c_{k}} + \rho_{cc(\z_{-k})}}}\right) -  \min\left(1, \frac{q_{-}(\z_{-k})}{ \frac{\rho_{c_{k}}}{\rho_{c_{k}} + \rho_{cc(\z_{-k})} }}\right),\\
&\qquad  \min\left(1, \frac{q_{+}(\z_{-k})}{ \frac{\rho_{c_{k}}}{\rho_{c_{k}} + \rho_{cc(\z_{-k})}}}\right) - \max\left(0, 1 - \frac{1-q_{-}(\z_{-k})}{  \frac{\rho_{c_{k}}}{\rho_{c_{k}} + \rho_{cc(\z_{-k})} }}\right)\Bigg).
\end{align*}

To simplify, recall that $\nu_k^{+}(\z_{-k})= \frac{1}{2}(\overline{D}_{k}(\z_{-k}, z_{k}^{+}) + 1) = \rho_{c_{k}} + \rho_{cc(\z_{-k})} + \rho_{a(\z_{-k})}$, $\nu_k^{-}(\z_{-k}) = \frac{1}{2}(\overline{D}_{k}(\z_{-k}, z_{k}^{-}) + 1) = \rho_{a(\z_{-k})}$, $\nu_k(\tilde{\z}_{-k}) =  \rho_{c_{k}}$, and $\nu_k(\z_{-k}) = \rho_k+ \rho_{cc(\z_{-k})}$.
We also see that we must have $1 - \nu_k^{+}(\z_{-k}) =  \rho_{n(\z_{-k})}$.

Simplifying the components of the bounds,
\begin{align*}
\frac{q_{+}(\z_{-k})}{ \frac{\rho_{c_{k}}}{\rho_{c_{k}} + \rho_{cc(\z_{-k})}}} &= \frac{(\rho_{c_{k}} + \rho_{cc(\z_{-k})} + \rho_{a(\z_{-k})})\overline{Y}(\z_{-k}, z_{k}^{+}|D_k(\z_{-k}, z_{k}^{+}) = +1)}{\nu_k(\tilde{\z}_{-k})}\\
&\qquad -  \frac{\rho_{a(\z_{-k})}\overline{Y}(\z_{-k}, z_{k}^{-}|D_k(\z_{-k}, z_{k}^{-}) = +1)}{ \nu_k(\tilde{\z}_{-k})} \\
&= \frac{\nu_k^{+}(\z_{-k})\overline{Y}(\z_{-k}, z_{k}^{+}|D_k(\z_{-k}, z_{k}^{+}) = +1) -\nu_k^{-}(\z_{-k})\overline{Y}(\z_{-k}, z_{k}^{-}|D_k(\z_{-k}, z_{k}^{-}) = +1)}{\nu_k(\tilde{\z}_{-k})}\\
&= \frac{\nu_k^{+}(\z_{-k})}{\nu_k(\tilde{\z}_{-k})}\overline{Y}(\z_{-k}, z_{k}^{+}|D_k(\z_{-k}, z_{k}^{+}) = +1) -\frac{\nu_k^{-}(\z_{-k})}{\nu_k(\tilde{\z}_{-k})}\overline{Y}_{a(\z_{-k})}(\z_{-k}, z_{k}^{-}),\\
\frac{q_{-}(\z_{-k})}{ \frac{\rho_{c_{k}}}{\rho_{c_{k}} + \rho_{cc(\z_{-k})} }} &= \frac{ (\rho_{c_{k}} + \rho_{cc(\z_{-k})} + \rho_{n(\z_{-k})}) \overline{Y}(\z_{-k}, z_{k}^{-}|D_k(\z_{-k}, z_{k}^{-}) = -1)}{\rho_{c_{k}}}\\
&\qquad - \frac{ \rho_{n(\z_{-k})}\overline{Y}(\z_{-k}, z_{k}^{+}|D_k(\z_{-k}, z_{k}^{+}) = -1)}{ \nu_k(\tilde{\z}_{-k})}\\
&= \frac{ (1-\nu_k^{-}(\z_{-k})) \overline{Y}(\z_{-k}, z_{k}^{-}|D_k(\z_{-k}, z_{k}^{-}) = -1)}{\nu_k(\tilde{\z}_{-k})}\\
&\qquad - \frac{(1 - \nu_k^{+}(\z_{-k}))\overline{Y}(\z_{-k}, z_{k}^{+}|D_k(\z_{-k}, z_{k}^{+}) = -1)}{ \nu_k(\tilde{\z}_{-k})}\\
&= \frac{ 1-\nu_k^{-}(\z_{-k})}{\nu_k(\tilde{\z}_{-k})} \overline{Y}(\z_{-k}, z_{k}^{-}|D_k(\z_{-k}, z_{k}^{-}) = -1)- \frac{1 - \nu_k^{+}(\z_{-k})}{\nu_k(\tilde{\z}_{-k})}\overline{Y}_{n(\z_{-k})}(\z_{-k}, z_{k}^{+}),\\
1 - \frac{1}{\frac{\rho_{c_{k}}}{\rho_{c_{k}} + \rho_{cc(\z_{-k})}}} &= - \frac{\rho_{cc(\z_{-k})}}{\rho_{c_{k}}} =  - \frac{\nu_k(\z_{-k}) -\nu_k(\tilde{\z}_{-k}) }{\nu_k(\tilde{\z}_{-k})}.
\end{align*}

We have lower bound
\begin{align*}
&\max\left(0, 1 - \frac{1-q_{+}(\z_{-k})}{ \frac{\rho_{c_{k}}}{\rho_{c_{k}} + \rho_{cc(\z_{-k})}}}\right) -  \min\left(1, \frac{q_{-}(\z_{-k})}{ \frac{\rho_{c_{k}}}{\rho_{c_{k}} + \rho_{cc(\z_{-k})} }}\right)\\
&= \max\left(0, \frac{\nu_k^{+}(\z_{-k})}{\nu_k(\tilde{\z}_{-k})}\overline{Y}(\z_{-k}, z_{k}^{+}|D_k(\z_{-k}, z_{k}^{+}) = +1) -\frac{\nu_k^{-}(\z_{-k})}{\nu_k(\tilde{\z}_{-k})}\overline{Y}_{a(\z_{-k})}(\z_{-k}, z_{k}^{-}) -  \frac{\nu_k(\z_{-k}) -\nu_k(\tilde{\z}_{-k}) }{\nu_k(\tilde{\z}_{-k})}\right)\\
&  -  \min\left(1, \frac{ 1-\nu_k^{-}(\z_{-k})}{ \nu_k(\tilde{\z}_{-k})} \overline{Y}(\z_{-k}, z_{k}^{-}|D_k(\z_{-k}, z_{k}^{-}) = -1) - \frac{1 - \nu_k^{+}(\z_{-k})}{ \nu_k(\tilde{\z}_{-k})}\overline{Y}_{n(\z_{-k})}(\z_{-k}, z_{k}^{+})\right)\\
\end{align*}
and upper bound
\begin{align*}
&\min\left(1, \frac{q_{+}(\z_{-k})}{ \frac{\rho_{c_{k}}}{\rho_{c_{k}} + \rho_{cc(\z_{-k})}}}\right) - \max\left(0, 1 - \frac{1-q_{-}(\z_{-k})}{  \frac{\rho_{c_{k}}}{\rho_{c_{k}} + \rho_{cc(\z_{-k})} }}\right)\\
&= \min\left(1,  \frac{\nu_k^{+}(\z_{-k})}{\nu_k(\tilde{\z}_{-k})}\overline{Y}(\z_{-k}, z_{k}^{+}|D_k(\z_{-k}, z_{k}^{+}) = +1) -\frac{\nu_k^{-}(\z_{-k})}{\nu_k(\tilde{\z}_{-k})}\overline{Y}_{a(\z_{-k})}(\z_{-k}, z_{k}^{-})\right)\\
& - \max\Bigg(0, \frac{ 1-\nu_k^{-}(\z_{-k})}{\nu_k(\tilde{\z}_{-k})} \overline{Y}(\z_{-k}, z_{k}^{-}|D_k(\z_{-k}, z_{k}^{-}) = -1)- \frac{1 - \nu_k^{+}(\z_{-k})}{\nu_k(\tilde{\z}_{-k})}\overline{Y}_{n(\z_{-k})}(\z_{-k}, z_{k}^{+})\\
&\qquad\qquad - \frac{\nu_k(\z_{-k}) -\nu_k(\tilde{\z}_{-k}) }{\nu_k(\tilde{\z}_{-k})}\Bigg).
\end{align*}

Leaving out the max and min arguments that depend on 0 and 1 in the bounds, we get the looser but still valid lower bound as
\begin{align*}
& \frac{\nu_k^{+}(\z_{-k})}{\nu_k(\tilde{\z}_{-k})}\overline{Y}(\z_{-k}, z_{k}^{+}|D_k(\z_{-k}, z_{k}^{+}) = +1) -\frac{\nu_k^{-}(\z_{-k})}{\nu_k(\tilde{\z}_{-k})}\overline{Y}_{a(\z_{-k})}(\z_{-k}, z_{k}^{-}) -  \frac{\nu_k(\z_{-k}) -\nu_k(\tilde{\z}_{-k}) }{\nu_k(\tilde{\z}_{-k})}\\
& \qquad -  \frac{ 1-\nu_k^{-}(\z_{-k})}{ \nu_k(\tilde{\z}_{-k})} \overline{Y}(\z_{-k}, z_{k}^{-}|D_k(\z_{-k}, z_{k}^{-}) = -1) + \frac{1 - \nu_k^{+}(\z_{-k})}{ \nu_k(\tilde{\z}_{-k})}\overline{Y}_{n(\z_{-k})}(\z_{-k}, z_{k}^{+})\\
&= \frac{1}{\nu_k(\tilde{\z}_{-k})}\overline{Y}(\z_{-k}, z_{k}^{+}) - \frac{ 1}{ \nu_k(\tilde{\z}_{-k})} \overline{Y}(\z_{-k}, z_{k}^{-}) -  \frac{\nu_k(\z_{-k}) -\nu_k(\tilde{\z}_{-k}) }{\nu_k(\tilde{\z}_{-k})}
\end{align*}
and upper bound as
\begin{align*}
&= \frac{\nu_k^{+}(\z_{-k})}{\nu_k(\tilde{\z}_{-k})}\overline{Y}(\z_{-k}, z_{k}^{+}|D_k(\z_{-k}, z_{k}^{+}) = +1) -\frac{\nu_k^{-}(\z_{-k})}{\nu_k(\tilde{\z}_{-k})}\overline{Y}_{a(\z_{-k})}(\z_{-k}, z_{k}^{-})\\
&\qquad - \frac{ 1-\nu_k^{-}(\z_{-k})}{\nu_k(\tilde{\z}_{-k})} \overline{Y}(\z_{-k}, z_{k}^{-}|D_k(\z_{-k}, z_{k}^{-}) = -1)+ \frac{1 - \nu_k^{+}(\z_{-k})}{\nu_k(\tilde{\z}_{-k})}\overline{Y}_{n(\z_{-k})}(\z_{-k}, z_{k}^{+})  + \frac{\nu_k(\z_{-k}) -\nu_k(\tilde{\z}_{-k}) }{\nu_k(\tilde{\z}_{-k})}\\
&= \frac{1}{\nu_k(\tilde{\z}_{-k})}\overline{Y}(\z_{-k}, z_{k}^{+}) - \frac{ 1}{\nu_k(\tilde{\z}_{-k})} \overline{Y}(\z_{-k}, z_{k}^{-})+ \frac{\nu_k(\z_{-k}) -\nu_k(\tilde{\z}_{-k}) }{\nu_k(\tilde{\z}_{-k})}.
\end{align*}

Either bounds can then be aggregated across the assignment of other factors, $\z_{-k}$, weighting by the total number of such assignments, $2^{-K}$.
 \end{proof}

\subsection{Proof of Remark~\ref{remark:conservative} }
\begin{proof}
First, we note that under Assumption~\ref{assump:weak_treat_ex}, because noncompliers have 0 effect of treatment assignment and all other units have effects bounded between $[-1,1]$,
\[-\sum_{\z_{-k} \in \mathcal{Z}_{-k}} \nu(\z_{-k})\leq \bm{g}_k^T\bm{Y} =\sum_{\z_{-k}}[\overline{Y}(\z_{-k}, z_{k}^{+}) -\overline{Y}(\z_{-k}, z_{k}^{-})] \leq \sum_{\z_{-k} \in \mathcal{Z}_{-k}} \nu(\z_{-k}) \]
Then, we can write the lower bound from Proposition~\ref{prop:weak_tr_ex} as
\begin{align*}
&\frac{1}{\rho_{c_{k}}}\left[\frac{1}{2^{K-1}}\bm{g}_k^T\bm{Y} - \frac{1}{2^{K-1}}\sum_{\z_{-k} \in \mathcal{Z}_{-k}} \nu(\z_{-k}) + \rho_{c_{k}}\right] - \frac{1}{t_k}\left[\frac{1}{2^{K-1}}\bm{g}_k^T\bm{Y} -  \frac{1}{2^{K-1}}\sum_{\z_{-k} \in \mathcal{Z}_{-k}} \nu(\z_{-k}) + t_k\right]\\
&= \frac{1}{\rho_{c_{k}}}\left[\frac{1}{2^{K-1}}\bm{g}_k^T\bm{Y} -  \frac{1}{2^{K-1}}\sum_{\z_{-k} \in \mathcal{Z}_{-k}} \nu(\z_{-k}) \right] - \frac{1}{t_k}\left[\frac{1}{2^{K-1}}\bm{g}_k^T\bm{Y} -  \frac{1}{2^{K-1}}\sum_{\z_{-k} \in \mathcal{Z}_{-k}} \nu(\z_{-k}) \right]\\
&=\underbrace{ \frac{t_k - \rho_{c_{k}}}{\rho_{c_{k}}t_k}}_{\leq 0}\underbrace{\left[\frac{1}{2^{K-1}}\bm{g}_k^T\bm{Y} -  \frac{1}{2^{K-1}}\sum_{\z_{-k} \in \mathcal{Z}_{-k}} \nu(\z_{-k}) \right] }_{\leq 0}\\
& \geq 0
\end{align*}
So the lower bound using $t_k$ is less than the lower bound using $\rho_{c_{k}}$.

Now consider the upper bound.
\begin{align*}
&\frac{1}{\rho_{c_{k}}}\left[\frac{1}{2^{K-1}}\bm{g}_k^T\bm{Y} +  \frac{1}{2^{K-1}}\sum_{\z_{-k} \in \mathcal{Z}_{-k}} \nu(\z_{-k}) - \rho_{c_{k}}\right] - \frac{1}{t_k}\left[\frac{1}{2^{K-1}}\bm{g}_k^T\bm{Y} +  \frac{1}{2^{K-1}}\sum_{\z_{-k} \in \mathcal{Z}_{-k}} \nu(\z_{-k}) - t_k\right]\\
&= \frac{1}{\rho_{c_{k}}}\left[\frac{1}{2^{K-1}}\bm{g}_k^T\bm{Y} +  \frac{1}{2^{K-1}}\sum_{\z_{-k} \in \mathcal{Z}_{-k}} \nu(\z_{-k}) \right] - \frac{1}{t_k}\left[\frac{1}{2^{K-1}}\bm{g}_k^T\bm{Y} + \frac{1}{2^K} \frac{1}{2^{K-1}}\sum_{\z_{-k} \in \mathcal{Z}_{-k}} \nu(\z_{-k}) \right]\\
&= \underbrace{\frac{t_k - \rho_{c_{k}}}{\rho_{c_{k}}t_k}}_{\leq 0} \underbrace{\left[\frac{1}{2^{K-1}}\bm{g}_k^T\bm{Y} +  \frac{1}{2^{K-1}}\sum_{\z_{-k} \in \mathcal{Z}_{-k}} \nu(\z_{-k}) \right] }_{\geq 0}\\
& \leq 0
\end{align*}
So the upper bound using $t_k$ is greater than the upper bound using $\rho_{c_{k}}$.
\end{proof}

\section{Additional interaction results}\label{append:int}
In Section~\ref{sec:int}, we consider estimating interactive effects for constant compliers on a single factor $k$.
However, as interactions inherently involve multiple factors, we may instead be interested in estimating the effect among units who are constant compliers on multiple factors.
We focus on two-factor interactions.
Thus, if we are estimating the interaction between factor $k$ and $\kp$, we are interested in identifying the effect among individuals who are constant compliers on both $k$ and $k'$.
This was addressed by \cite{blackwell2021noncompliance} under the treatment exclusion restriction assumption on all factors.
We introduce here two weaker assumptions that can be used to obtain bounds.

Throughout this section, we will use notation $\z_{-(k, \kp)}$ to indicate the assignment vector $\z$ with the assignments for factor $k$ and $\kp$ removed.
We define $D_{i,k\circ\kp}(\z) = D_{i,k}(\z)D_{i,\kp}(\z)$.
Similarly, $\bm{g}_{k\circ\kp}$ is the contrast vector for the two-factor interaction between factors $k$ and $\kp$, which is defined by the elemnt-wise (Hadamard) product between $\bm{g}_{k}$ and $\bm{g}_{\kp}$.
Also let $c_{k,\kp}$ be the subscript denoting constant compliers on factors $k$ and $\kp$ so that, e.g., $\rho_{c_{k,\kp}}$ is the proportion of units who are constant compliers on both  $k$ and $\kp$.
The effect we wish to estimate is then
  \begin{align*}
    \delta_{k\kp}& \equiv \frac{1}{2^{K-1}} \sum_{\z_{-(k, \kp)}}\Big(\overline{Y}_{c_{k\kp}}(\z_{-(k, \kp)}, z_{k}^{+}, z_{\kp}^{+}) - \overline{Y}_{c_{k\kp}}(\z_{-(k, \kp)}, z_{k}^{-}, z_{\kp}^{+}) \\
    &\qquad \qquad \qquad \qquad - \left[ \overline{Y}_{c_{k\kp}}(\z_{-(k, \kp)}, z_{k}^{+}, z_{\kp}^{-}) - \overline{Y}_{c_{k\kp}}(\z_{-(k, \kp)}, z_{k}^{-}, z_{\kp}^{-})\right]\Big).
    \end{align*}

Further, define the interactive effect of treatment assignment on uptake as
  \begin{align*}
   \nu_{k, \kp}(\tilde{\z}_{-(k, \kp)})&=
\frac{1}{4}\Big(\overline{D}_{k\circ\kp}(\z_{-(k, \kp)}, z_{k}^{+}, z_{\kp}^{+}) - \overline{D}_{k\circ\kp}(\z_{-(k, \kp)}, z_{k}^{-}, z_{\kp}^{+})\\
&\qquad \qquad - \left[ \overline{D}_{k\circ\kp}(\z_{-(k, \kp)}, z_{k}^{+}, z_{\kp}^{-}) - \overline{D}_{k\circ\kp}(\z_{-(k, \kp)}, z_{k}^{-}, z_{\kp}^{-})\right]\Big)
 \end{align*}

To provide identification results for this effect, first, we need the following assumption which is stronger than Assumption~\ref{assump:order_comp}.
  \begin{assumption}[Least compliant profile for factors $k$ and $k'$]\label{assump:order_comp_joint}
There is a $\tilde{\z}_{-(k,\kp)} \in \mathcal{Z}_{-(k,\kp)}$, such that 
\begin{align*}
&D_{i,k\circ\kp}(\z_{-(k, \kp)}, z_{k}^{+}, z_{\kp}^{+}) - D_{i,k\circ\kp}(\z_{-(k, \kp)}, z_{k}^{-}, z_{\kp}^{+}) - \left[ D_{i,k\circ\kp}(\z_{-(k, \kp)}, z_{k}^{+}, z_{\kp}^{-}) - D_{i,k\circ\kp}(\z_{-(k, \kp)}, z_{k}^{-}, z_{\kp}^{-})\right]\\
& \geq D_{i,k\circ\kp}(\tilde{\z}_{-(k, \kp)}, z_{k}^{+}, z_{\kp}^{+}) - D_{i,k\circ\kp}(\tilde{\z}_{-(k, \kp)}, z_{k}^{-}, z_{\kp}^{+}) - \left[ D_{i,k\circ\kp}(\tilde{\z}_{-(k, \kp)}, z_{k}^{+}, z_{\kp}^{-}) - D_{i,k\circ\kp}(\tilde{\z}_{-(k, \kp)}, z_{k}^{-}, z_{\kp}^{-})\right]
\end{align*}
 and for all $i$.
\end{assumption}

Second, we need the following assumption that controls the interactive effect of assignment on compliance between the two factors under consideration.
\begin{assumption}[Conditional treatment exclusion restriction between factors $k$ and $\kp$]\label{assump:strong_cond_joint}
There is a conditional treatment exclusion between factors $k$ and $\kp$ if
\[D_{i,k}(\z_{-(k, \kp)}, z_{k} = z, z_{\kp}^{+}) = D_{i,k}(\z_{-(k, \kp)}, z_{k} = z, z_{\kp}^{-})\]
and 
\[D_{i,\kp}(\z_{-(k, \kp)}, z_{k}^+ z, z_{\kp}=z) = D_{i,\kp}(\z_{-(k, \kp)}, z_{k}^-, z_{\kp} = z)\]
for all $z \in \{-1, +1\}$, $\z_{-(k, \kp)}$, and $i$.
In other words, conditional on the assignment of other factors, the uptake on factor $k$ does not depend on the assignment of factor $\kp$ and vice versa. 
\end{assumption}
Assumption~\ref{assump:strong_cond_joint}  is weaker than the treatment exclusion restriction assumption of \cite{blackwell2021noncompliance} because it is (i) only required on the two factors under consideration (ii) conditional and thus allows compliance to change across assignments for the other $K-2$ factors.

\begin{proposition}\label{prop:append_int}
Under Assumptions~\ref{assump:er}, \ref{assump:mono}, \ref{assump:weak_treat_ex}, \ref{assump:order_comp_joint}, and \ref{assump:strong_cond_joint},  $\delta_{k\kp} \in  [\widetilde{\delta}_{k, \kp}-b_{k,\kp},\widetilde{\delta}_{k, \kp} + b_{k,\kp}]  \cap [-1,1]$ where
\[\widetilde{\delta}_{k, \kp} = \frac{ \frac{1}{2^{K-1}}\bm{g}_{k\circ\kp}^{T}\bm{Y}}{\nu_{k, \kp}(\tilde{\z}_{-(k, \kp)})}\]
and
\[b_{k,\kp} = \frac{1}{\nu_{k, \kp}(\tilde{\z}_{-(k, \kp)})} \left[\frac{1}{2^{K}}\bm{g}_{k\circ\kp}^{T}\bm{D}_{k\kp}- \nu_{k, \kp}(\tilde{\z}_{-(k, \kp)})\right]\]
 
\end{proposition}

\begin{proof}
We start by showing identification of the proportion of constant compliers for factors $k$ and $\kp$ using Assumptions~\ref{assump:mono}, \ref{assump:order_comp_joint}, and \ref{assump:strong_cond_joint}.
We have the conditional unit-level interactive effect for factor $k$ and $\kp$ on uptake is
\begin{align*}
&D_{i,k\circ\kp}(\z_{-(k, \kp)}, z_{k}^{+}, z_{\kp}^{+}) - D_{i,k\circ\kp}(\z_{-(k, \kp)}, z_{k}^{-}, z_{\kp}^{+}) - \left[ D_{i,k\circ\kp}(\z_{-(k, \kp)}, z_{k}^{+}, z_{\kp}^{-}) - D_{i,k\circ\kp}(\z_{-(k, \kp)}, z_{k}^{-}, z_{\kp}^{-})\right]\\
&=D_{i,k}(\z_{-(k, \kp)}, z_{k}^{+}, z_{\kp}^{+})D_{i,\kp}(\z_{-(k, \kp)}, z_{k}^{+}, z_{\kp}^{+}) - D_{i,k}(\z_{-(k, \kp)}, z_{k}^{-}, z_{\kp}^{+}) D_{i,k\kp}(\z_{-(k, \kp)}, z_{k}^{-}, z_{\kp}^{+})\\
&\quad - \left[ D_{i,k}(\z_{-(k, \kp)}, z_{k}^{+}, z_{\kp}^{-})D_{i,\kp}(\z_{-(k, \kp)}, z_{k}^{+}, z_{\kp}^{-}) - D_{i,k}(\z_{-(k, \kp)}, z_{k}^{-}, z_{\kp}^{-})D_{i,\kp}(\z_{-(k, \kp)}, z_{k}^{-}, z_{\kp}^{-})\right]\\
&=D_{i,k}(\z_{-(k, \kp)}, z_{k}^{+}, z_{\kp}^{+})\left[D_{i,\kp}(\z_{-(k, \kp)}, z_{k}^{+}, z_{\kp}^{+}) -D_{i,\kp}(\z_{-(k, \kp)}, z_{k}^{+}, z_{\kp}^{-})\right]\\
&\quad- D_{i,k}(\z_{-(k, \kp)}, z_{k}^{-}, z_{\kp}^{+}) \left[D_{i,\kp}(\z_{-(k, \kp)}, z_{k}^{-}, z_{\kp}^{+})-D_{i,\kp}(\z_{-(k, \kp)}, z_{k}^{-}, z_{\kp}^{-})\right]\\
&=\left[D_{i,k}(\z_{-(k, \kp)}, z_{k}^{+}, z_{\kp}^{+})- D_{i,k}(\z_{-(k, \kp)}, z_{k}^{-}, z_{\kp}^{+})\right]\left[D_{i,\kp}(\z_{-(k, \kp)}, z_{k}^{+}, z_{\kp}^{+}) -D_{i,\kp}(\z_{-(k, \kp)}, z_{k}^{+}, z_{\kp}^{-})\right]\\
&=4C_{i,k, \kp} + 4CC_{i,k, \kp}(\z_{-(k, \kp)}),
\end{align*}
where $C_{i,k, \kp}$ is the indicator that unit $i$ is a constant complier and $CC_{i,k, \kp}(\z_{-(k, \kp)})$ is the indicators that unit $i$ is a conditional complier, conditional on $\z_{-(k, \kp)}$, for factors $k$ and $\kp$.
The second equality comes from the definition of $D_{i,k\circ\kp}(\z)$, the third and fourth equalities come from Assumption~\ref{assump:strong_cond_joint}, and the last follows from Assumption~\ref{assump:mono}, and the fact that only conditional compliers and constant compliers on both factors will have 
\[\left[D_{i,k}(\z_{-(k, \kp)}, z_{k}^{+}, z_{\kp}^{+})- D_{i,k}(\z_{-(k, \kp)}, z_{k}^{-}, z_{\kp}^{+})\right]\left[D_{i,\kp}(\z_{-(k, \kp)}, z_{k}^{+}, z_{\kp}^{+}) -D_{i,\kp}(\z_{-(k, \kp)}, z_{k}^{+}, z_{\kp}^{-})\right] = 4\]
 (and conditional noncompliers will have it equal to 0).
 
 Thus, we must have the average conditional interactive effect for factor $k$ and $\kp$ on uptake, $\nu_{k, \kp}(\z_{-(k, \kp)})$, is
 \begin{align*}
 &4\nu_{k, \kp}(\z_{-(k, \kp)})=\\
&\overline{D}_{k\circ\kp}(\z_{-(k, \kp)}, z_{k}^{+}, z_{\kp}^{+}) - \overline{D}_{k\circ\kp}(\z_{-(k, \kp)}, z_{k}^{-}, z_{\kp}^{+}) - \left[ \overline{D}_{k\circ\kp}(\z_{-(k, \kp)}, z_{k}^{+}, z_{\kp}^{-}) - \overline{D}_{k\circ\kp}(\z_{-(k, \kp)}, z_{k}^{-}, z_{\kp}^{-})\right]\\
&=4\rho_{c_{k,\kp}} + 4\rho_{cc(\z_{-(k, \kp)})}.
\end{align*}

Further, from Assumption~\ref{assump:order_comp_joint}, we must have $\nu_{k, \kp}(\tilde{\z}_{-(k, \kp)}) = \rho_{c_{k,\kp}}$.

Next we show the bounds.
 We have the conditional intent-to-treat interactive effect between factors $k$ and $\kp$ can be decomposed as
  \begin{align*}
   & \gamma_{k,\kp}(\z_{-(k,\kp)})\\
     &\equiv \overline{Y}(\z_{-(k,\kp)}, z_{k}^{+}, z_{\kp}^{+}) - \overline{Y}(\z_{-(k,\kp)}, z_{k}^{-}, z_{\kp}^{+}) - \left[\overline{Y}(\z_{-(k,\kp)}, z_{k}^{+}, z_{\kp}^{-}) - \overline{Y}(\z_{-(k,\kp)}, z_{k}^{-}, z_{\kp}^{-})\right] \\
                      &= \rho_{c_{k, \kp}}\left[\overline{Y}_{c_{k, \kp}}(\z_{-(k,\kp)}, z_{k}^{+}, z_{\kp}^{+}) - \overline{Y}_{c_{k, \kp}}(\z_{-(k,\kp)}, z_{k}^{-}, z_{\kp}^{+}) - \left[\overline{Y}_{c_{k, \kp}}(\z_{-(k,\kp)}, z_{k}^{+}, z_{\kp}^{-}) - \overline{Y}_{c_{k, \kp}}(\z_{-(k,\kp)}, z_{k}^{-}, z_{\kp}^{-})\right]\right] \\ 
                      &\quad + \rho_{cc(\z_{-(k,\kp)})}\Big[\overline{Y}_{cc(\z_{-(k,\kp)})}(\z_{-(k,\kp)}, z_{k}^{+}, z_{\kp}^{+}) - \overline{Y}_{cc(\z_{-(k,\kp)})}(\z_{-(k,\kp)}, z_{k}^{-}, z_{\kp}^{+}) \\
                      &\qquad \qquad  \qquad  \qquad - \left[\overline{Y}_{cc(\z_{-(k,\kp)})}(\z_{-(k,\kp)}, z_{k}^{+}, z_{\kp}^{-}) - \overline{Y}_{cc(\z_{-(k,\kp)})}(\z_{-(k,\kp)}, z_{k}^{-}, z_{\kp}^{-})\right]\Big] \\
    & \quad  + \rho_{cn(\z_{-(k,\kp)})}\Big[\overline{Y}_{cn(\z_{-(k,\kp)})}(\z_{-(k,\kp)}, z_{k}^{+}, z_{\kp}^{+}) - \overline{Y}_{cn(\z_{-(k,\kp)})}(\z_{-(k,\kp)}, z_{k}^{-}, z_{\kp}^{+}) \\
                      &\qquad \qquad  \qquad  \qquad - \left[\overline{Y}_{cn(\z_{-(k,\kp)})}(\z_{-(k,\kp)}, z_{k}^{+}, z_{\kp}^{-}) - \overline{Y}_{cn(\z_{-(k,\kp)})}(\z_{-(k,\kp)}, z_{k}^{-}, z_{\kp}^{-})\right]\Big].
  \end{align*}
  Consider a unit $i$ who is a conditional noncomplier on either $k$ or $\kp$, conditional on $\z_{-(k,\kp)}$. 
  Without loss of generality, take unit $i$ to be a noncomplier on factor $k$.
  Then under Assumptions~\ref{assump:er}, and \ref{assump:mono}, and \ref{assump:strong_cond_joint},
    \begin{align*}
   & Y_{i}(\z_{-(k,\kp)}, z_{k}^{+}, z_{\kp}^{+}) - Y_{i}(\z_{-(k,\kp)}, z_{k}^{-}, z_{\kp}^{+}) - \left[Y_{i}(\z_{-(k,\kp)}, z_{k}^{+}, z_{\kp}^{-}) - Y_{i}(\z_{-(k,\kp)}, z_{k}^{-}, z_{\kp}^{-})\right]\\
    &=0.
    \end{align*}
    This follows because under Assumptions~\ref{assump:mono} and \ref{assump:strong_cond_joint}, for this unit $D_{ik}(\z_{-(k,\kp)}, z_{k}^{+}, z_{\kp}) = D_{ik}(\z_{-(k,\kp)}, z_{k}^{+}, z_{\kp})$ for $z_{\kp} \in \{z_{\kp}^{+}, z_{\kp}^{-}\}$.
Then under Assumption~\ref{assump:weak_treat_ex}, we must have $Y_{i}(\z_{-(k,\kp)}, z_{k}^{+}, z_{\kp}^{+}) - Y_{i}(\z_{-(k,\kp)}, z_{k}^{-}, z_{\kp}^{+}) = 0$ and $Y_{i}(\z_{-(k,\kp)}, z_{k}^{+}, z_{\kp}^{-}) - Y_{i}(\z_{-(k,\kp)}, z_{k}^{-}, z_{\kp}^{-}) = 0$.
Thus, we have
\begin{align*}
   & \gamma_{k,\kp}(\z_{-(k,\kp)})\\
                       &= \rho_{c_{k, \kp}}\left[\overline{Y}_{c_{k, \kp}}(\z_{-(k,\kp)}, z_{k}^{+}, z_{\kp}^{+}) - \overline{Y}_{c_{k, \kp}}(\z_{-(k,\kp)}, z_{k}^{-}, z_{\kp}^{+}) - \left[\overline{Y}_{c_{k, \kp}}(\z_{-(k,\kp)}, z_{k}^{+}, z_{\kp}^{-}) - \overline{Y}_{c_{k, \kp}}(\z_{-(k,\kp)}, z_{k}^{-}, z_{\kp}^{-})\right]\right] \\ 
                     &\quad + \rho_{cc(\z_{-(k,\kp)})}\Big[\overline{Y}_{cc(\z_{-(k,\kp)})}(\z_{-(k,\kp)}, z_{k}^{+}, z_{\kp}^{+}) - \overline{Y}_{cc(\z_{-(k,\kp)})}(\z_{-(k,\kp)}, z_{k}^{-}, z_{\kp}^{+}) \\
                     &\qquad \qquad  \qquad  \qquad - \left[\overline{Y}_{cc(\z_{-(k,\kp)})}(\z_{-(k,\kp)}, z_{k}^{+}, z_{\kp}^{-}) - \overline{Y}_{cc(\z_{-(k,\kp)})}(\z_{-(k,\kp)}, z_{k}^{-}, z_{\kp}^{-})\right]\Big].
  \end{align*}
  Rearranging and bounding the conditional complier term gives us
  \begin{align*}
                      &  \rho_{c_{k, \kp}}\left[\overline{Y}_{c_{k, \kp}}(\z_{-(k,\kp)}, z_{k}^{+}, z_{\kp}^{+}) - \overline{Y}_{c_{k, \kp}}(\z_{-(k,\kp)}, z_{k}^{-}, z_{\kp}^{+}) - \left[\overline{Y}_{c_{k, \kp}}(\z_{-(k,\kp)}, z_{k}^{+}, z_{\kp}^{-}) - \overline{Y}_{c_{k, \kp}}(\z_{-(k,\kp)}, z_{k}^{-}, z_{\kp}^{-})\right]\right] \\ 
                     &= \gamma_{k,\kp}(\z_{-(k,\kp)}) - \rho_{cc(\z_{-(k,\kp)})}\Big[\overline{Y}_{cc(\z_{-(k,\kp)})}(\z_{-(k,\kp)}, z_{k}^{+}, z_{\kp}^{+}) - \overline{Y}_{cc(\z_{-(k,\kp)})}(\z_{-(k,\kp)}, z_{k}^{-}, z_{\kp}^{+}) \\
                     &\qquad \qquad  \qquad  \qquad \qquad  \qquad - \left[\overline{Y}_{cc(\z_{-(k,\kp)})}(\z_{-(k,\kp)}, z_{k}^{+}, z_{\kp}^{-}) - \overline{Y}_{cc(\z_{-(k,\kp)})}(\z_{-(k,\kp)}, z_{k}^{-}, z_{\kp}^{-})\right]\Big]\\
                     &\in \left(\gamma_{k,\kp}(\z_{-(k,\kp)}) -2\rho_{cc(\z_{-(k,\kp)})}, \gamma_{k,\kp}(\z_{-(k,\kp)}) +2\rho_{cc(\z_{-(k,\kp)})}  \right)\\
                     &= \left(\gamma_{k,\kp}(\z_{-(k,\kp)}) -2[\nu_{k, \kp}(\z_{-(k, \kp)}) - \nu_{k, \kp}(\tilde{\z}_{-(k, \kp)})], \gamma_{k,\kp}(\z_{-(k,\kp)}) +2[\nu_{k, \kp}(\z_{-(k, \kp)}) - \nu_{k, \kp}(\tilde{\z}_{-(k, \kp)})]\right).
  \end{align*}
 Aggregating across $\z_{-k, \kp}$, note that 
 \[\frac{1}{2^{K-1}}\bm{g}_{k\circ\kp}^{T}\bm{Y}=\frac{1}{2^{K-1}}\sum_{\z_{-(k,\kp)}}\gamma_{k,\kp}(\z_{-(k,\kp)})\]
 and
 \[\frac{1}{2^{K}}\bm{g}_{k\circ\kp}^{T}\bm{D}_{k\kp}=\frac{1}{2^{K-2}}\sum_{\z_{-(k,\kp)}}\nu_{k, \kp}(\z_{-(k, \kp)}).\]
 So then we have the result
  \begin{align*}
    \delta_{k\kp}
  &\in\Bigg( \max \left\{0, \frac{1}{\rho_{c_{k, \kp}}} \left[\frac{1}{2^{K-1}}\bm{g}_{k\circ\kp}^{T}\bm{Y} -\frac{1}{2^{K}}\bm{g}_{k\circ\kp}^{T}\bm{D}_{k\kp}+ \nu_{k, \kp}(\tilde{\z}_{-(k, \kp)})\right] \right\}, \\
  &\qquad
  \min \left\{1, \frac{1}{\rho_{c_{k, \kp}}} \left[\frac{1}{2^{K-1}}\bm{g}_{k\circ\kp}^{T}\bm{Y} +\frac{1}{2^{K}}\bm{g}_{k\circ\kp}^{T}\bm{D}_{k\kp}- \nu_{k, \kp}(\tilde{\z}_{-(k, \kp)})\right] \right\}\Bigg).
  \end{align*}
\end{proof}

\subsection{Proof of bounds on interaction}
We prove the bounds for interactions involving factor $k$ under Assumptions~\ref{assump:cr}, \ref{assump:er}, \ref{assump:mono}, \ref{assump:order_comp}, and~\ref{assump:weak_treat_ex}.
\begin{proof}
Define $\tilde{\bm{g}}_{f(k)}(\z) \in \{-1, 1\}$ give the sign of $\bm{g}_{f(k)}$ for each treatment $\z$.
We have $\delta_{f(k)}$ as 
   \begin{align*}
     \delta_{f(k)} =& \frac{1}{2^{K-1}\nu_k(\tilde{\z}_{-k})} \Bigg[ \bm{g}_{f(k)}^{T}\bm{Y} \\
   & \hspace{2.8cm}  - \sum_{\z_{-k} \in \mathcal{Z}_{-k}} \tilde{\bm{g}}_{f(k)}(z^{+}_k, \z_{-k})(\nu_k(\z_{-k}) - \nu_k(\tilde{\z}_{-k}))\left[\overline{Y}_{cc(\z_{-k})}(\z_{-k}, z_{k}^{+}) - \overline{Y}_{cc(\z_{-k})}(\z_{-k}, z_{k}^{-})\right]
  \\
 & \hspace{2.8cm} - \sum_{\z_{-k} \in \mathcal{Z}_{-k}} \tilde{\bm{g}}_{f(k)}(z^{+}_k, \z_{-k})(1 - \nu_k(\z_{-k})) \left[\overline{Y}_{cn(\z_{-k})}(\z_{-k}, z_{k}^{+}) - \overline{Y}_{cn(\z_{-k})}(\z_{-k}, z_{k}^{-})\right] \Bigg]\\
  =& \frac{1}{2^{K-1}\nu_k(\tilde{\z}_{-k})} \Bigg[ \bm{g}_{f(k)}^{T}\bm{Y} \\
   & \hspace{2.8cm}  - \sum_{\z_{-k} \in \mathcal{Z}_{-k}} \tilde{\bm{g}}_{f(k)}(z^{+}_k, \z_{-k})(\nu_k(\z_{-k}) - \nu_k(\tilde{\z}_{-k}))\left[\overline{Y}_{cc(\z_{-k})}(\z_{-k}, z_{k}^{+}) - \overline{Y}_{cc(\z_{-k})}(\z_{-k}, z_{k}^{-})\right].
 \end{align*}
 As in the proof of Proposition~\ref{prop:weak_tr_ex}, the last line holds because under the assumptions, in particular Assumption~\ref{assump:weak_treat_ex}, the effect among conditional noncompliers is 0.
We then bound the effect of conditional compliers in the interval $[-1,1]$.
\end{proof}

\section{Sharp bounds}\label{append:sharp_bounds}

In special cases, it is possible to derive sharp bounds for conditional effects of a factor for constant compliers. In this appendix, we focus on the case with a binary outcome, where we can apply the linear programming approach to sharp bounds derived by \citet{BalPea97}. In this approach, we treat the joint distribution of the principal strata for compliance and the outcome as a set of unknown parameters and view the observational distribution of the data as a set of constraints. We can then find the values of the unknown principal strata parameters that maximize or minimize the quantity of interest subject to these constraints and any other constraints imposed by the maintained assumptions. \citet{BlaBroHil25} describes this approach more fully. 

For ease of exposition, we adopt a superpopulation approach to illustrate these bounds and focus on the $K=2$ case. Define the following transformations of the treatment assignment and uptake vectors:
$$
\begin{aligned}
  D^{\ast}_{i} = \begin{cases}
    1 &\text{if } \bm{D}_{i} = (1, 1) \\
    2 &\text{if } \bm{D}_{i} = (-1, 1) \\
    3 &\text{if } \bm{D}_{i} = (1, -1) \\
    4 &\text{if } \bm{D}_{i} = (-1, -1) \\
  \end{cases},
  \qquad
    Z^{\ast}_{i} = \begin{cases}
    1 &\text{if } \bm{Z}_{i} = (1, 1) \\
    2 &\text{if } \bm{Z}_{i} = (-1, 1) \\
    3 &\text{if } \bm{Z}_{i} = (1, -1) \\
    4 &\text{if } \bm{Z}_{i} = (-1, -1) \\
  \end{cases}.
\end{aligned}
$$ 

With these in hand, we can describe the distribution of the observed data as
$$
p_{yd|z} = \P(Y_i = y, D^{\ast}_i = d \mid Z^{\ast}_i = z).
$$
The sharp bounds for the factorial effect for the constant compliers under all the assumption of the main text and one-sided noncompliance would be:
\begin{align}
  \delta_{k} &\leq \frac{1}{2}\min \left\{\begin{matrix}
  p_{11|1} + p_{12|1} + p_{13|1} - p_{03|1} + p_{14|1} - p_{12|2} - p_{14|2} + p_{13|3} + p_{03|3} + p_{14|3} - p_{14|4} \\
  2p_{11|1} + p_{01|1} + p_{12|1} + 2p_{13|1} + 2p_{14|1} - p_{12|2} - 2p_{14|2} \\  
p_{11|1} + p_{12|1} + p_{13|1} + 2p_{14|1} - p_{12|2} - 2p_{14|2} + p_{13|3} \\
-p_{01|1} + p_{12|1} - p_{03|1} + p_{14|1} - p_{12|2} - p_{14|2} + 2p_{13|3} + p_{03|3} + p_{14|3} - p_{14|4}
  \end{matrix}
  \right\} \\
  \delta_{k} & \geq \frac{1}{2}\max \left\{\begin{matrix}
    p_{11|1} + 2p_{12|1} + p_{02|1} + 2p_{13|1} + 2p_{14|1} - 2p_{12|2} - p_{02|2} - 2p_{14|2}\\                     
2p_{11|1} + p_{01|1} + 2p_{12|1} + p_{02|1} + 2p_{13|1} + p_{03|1} + 2p_{14|1} - 2p_{12|2} - p_{02|2} - 2p_{14|2} - p_{03|3} \\ 
    p_{11|1} + p_{12|1} + 2p_{13|1} + p_{14|1} - p_{12|2} - p_{14|2} + p_{14|3} - p_{14|4} \\
2p_{11|1} + p_{01|1} + p_{12|1} + 2p_{13|1} + p_{03|1} + p_{14|1} - p_{12|2} - p_{14|2} - p_{03|3} + p_{14|3} - p_{14|4} 
  \end{matrix}
  \right\}.
\end{align}

Because the linear programming approach to deriving these bounds uses all the restriction implied by the maintained assumptions, these bounds are sharp. We can also derive bounds for the two-sided noncompliance case, though these are considerably more complicated. To estimate these bounds, one can plug in in the sample averages for these population quantities and use the approach of \citet{ImbMan04} to obtain confidence intervals. 

\section{Simulation}\label{append:simulation}

We now present a simulation study to illustrate the features of our bounds. We focus on designs with $K=2$ and $K=5$ to understand how increasing the number of factors affects the bounds. In our setting, there is only noncompliance on the first factor, $k=1$ and we are interested in the factorial effect of this factor.  We assume there are three possible compliance types:
\begin{itemize}
  \item perfect compliers ($c_{1}$) for factor 1: $D_{i1}(z_{1}, \z_{-1}) = z_{1}$
  \item never-takers for the ``most active'' profile, compliers for all other profiles $cc_{1}$: $D_{i1}(z_{1}, \mathbf{1}^{K-1}) = -1$ and $D_{i1}(z_{1}, \z_{-1}) = z_{1}$ for all $\z_{-1} \neq \mathbf{1}^{K-1}$
  \item never-takers for the top 2 ``most active'' profiles ($cc_{2}$): $D_{i1}(z_{1}, \mathbf{1}^{K-1}) = -1$, $D_{i1}(z_{1}, -1, \mathbf{1}^{K-2}) = -1$,  and $D_{i1}(z_{1}, \z_{-1}) = z_{1}$ for all $\z_{-1} \notin \{\mathbf{1}^{K-1}, (-1, \mathbf{1}^{K-2})\}$  
\end{itemize}
In this case, the least compliant profile is $\widetilde{\z}_{-k} = \mathbf{1}^{K-1}$. To see how the bounds vary as a function of compliance, we vary the proportion of perfect compliers $\rho_{c_{1}} \in \{0.2, 0.4, 0.8, 0.9\}$. We consider two ways to distribute the noncompliance between the two conditional complier groups:
\begin{itemize}
  \item Scenario 1: only noncompliance on $cc_{1}$: $\rho_{cc_{1}} = 1 - \rho_{c_{1}}$, $\rho_{cc_{2}} = 0$.
  \item Scenario 2: two noncompliant profiles, with one only marginally larger than the other $\rho_{cc_{1}} = 0.01$, $\rho_{cc_{2}} = 1 - 0.01 - \rho_{c_{1}}$.
\end{itemize}
The first setting has only one noncompliant profile and so it is clear how to select it in deriving the bounds. The second is more ambiguous, since the least compliant group is only marginally less compliant than the next least compliant profile. In different samples, either one might be the observed least compliant, allowing us to test the bounding approach in Remark~\ref{remark:conservative} from the main text, where we empirically estimate the least compliant profile.

Our data generating process is $\P(Y_{i}(\z) = 1)  = 0.05 + 0.2 \times \mathbb{I}(z_{1} = 1)\times C_{i1}$, so that the factorial effect of factor 1 is 0.2 for the perfect compliers and all other effects are 0. We sample the compliance types with above probabilities, conduct a complete randomization for treatment assignment, $\bm{Z}_{i}$, and use the compliance types to generate treatment uptake, $\bm{D}_{i}$. We use a sample size $N=1600$ and calculate the average upper and lower bounds along with the average width across 1,000 simulation iterations.

\begin{table}[t]
  \centering
\caption{Scenario 1 simulation results with only one noncompliant profile ($N=1600$) \label{tab:main_sim_results}}

\begin{tabular}{rrrrr}
\toprule
Num. Factors & Perfect Compliance Prob & Lower (avg.) & Upper (avg.) & Width (avg.)\\
\midrule
2 & 0.2 & -0.026 & 0.325 & 0.351\\
2 & 0.4 & 0.038 & 0.263 & 0.226\\
2 & 0.8 & 0.077 & 0.230 & 0.153\\
2 & 0.9 & 0.145 & 0.228 & 0.084\\
\addlinespace
5 & 0.2 & -0.243 & 0.454 & 0.696\\
5 & 0.4 & -0.108 & 0.326 & 0.433\\
5 & 0.8 & -0.006 & 0.261 & 0.267\\
5 & 0.9 & 0.094 & 0.248 & 0.154\\
\bottomrule
\end{tabular}
  \end{table}

The results of scenario 1 are in Table~\ref{tab:main_sim_results}. We find that the width and thus informativeness of the bounds shrink as we increase the proportion of perfect compliers. Conversely, as we increase the number of factors, we see an increase in the width of the bounds, even though we are focused on a single factor's effect.  In this data-generating process, we find that the bounds are informative (do not contain 0) when there are only when perfect compliance is at or above 40\% for 2 factors and when perfect compliance is above 80\% for 5 factors.

\begin{table}[t]
  \centering
\caption{Scenario 2 simulation results with two noncompliant profiles with similar sizes and either using the known least compliant profile or estimating it from the data ($K=5, N=1600$) \label{tab:ambig_results}}
\small

\begin{tabular}{rlrrr}
\toprule
Perfect Compliance Prob & Least Compliant Profile & Lower & Upper & Width\\
\midrule
0.2 & Known & -0.22 & 0.44 & 0.66\\
0.4 & Known & -0.09 & 0.31 & 0.40\\
0.8 & Known & 0.01 & 0.26 & 0.25\\
0.9 & Known & 0.10 & 0.24 & 0.14\\
\addlinespace
0.2 & Estimated & -0.26 & 0.50 & 0.76\\
0.4 & Estimated & -0.11 & 0.35 & 0.45\\
0.8 & Estimated & -0.01 & 0.27 & 0.28\\
0.9 & Estimated & 0.08 & 0.25 & 0.17\\
\bottomrule
\end{tabular}
  \end{table}
\input{}

The results for scenario 2 are in Table~\ref{tab:ambig_results} and they confirm the arguments of Remark~\ref{remark:conservative}. When the least compliant profile is estimated rather than known, the estimated bounds are 15--20\% wider than the bounds with the known least compliant profile, at least for this data-generating process. 

\section{Identification of interaction effects for constant compliers}\label{sec:int}
Any factorial effect involving factor $k$ can be written as a linear combination of the conditional effects $\overline{Y}(\z_{-k}, z_{k}^{+}) - \overline{Y}(\z_{-k}, z_{k}^{-})$.
Therefore, we can use the same identification strategy as above to estimate the interaction effect among constant compliers for factor $k$, for any interaction involving factor $k$.
In particular, consider estimating some interaction with factor $k$, called $\delta_{f(k)}$ which has contrast vector $\bm{g}_{f(k)}$, among constant compliers.
Following the logic above under assumptions~\ref{assump:cr}, \ref{assump:er}, \ref{assump:mono}, \ref{assump:order_comp}, and~\ref{assump:weak_treat_ex} for factor $k$, we have $\delta_{f(k)} \in[ \delta^{\prime}_{f(k)}-\tilde{b}_{k}, \delta^{\prime}_{f(k)} + \tilde{b}_{k}] \cap [-1, 1]$, where
    \begin{align*}      
    \delta^{\prime}_{f(k)} & =  \frac{1}{2^{K-1}\nu_k(\tilde{\z}_{-k})} \bm{g}_{f(k)}^{T}\bm{Y}  .
     \end{align*}

Here, we would end up with a different estimand for the same interaction for the constant compliers on factors involved which may not be desirable.
For example, if we are interested in the interaction between factor 1 and 2, we can use the above bound to get this effect among constant compliers on factor 1 or constant compliers on factor 2 -- these correspond to two different estimands given the two different groups.
Under the treatment exclusion restriction assumption of \cite{blackwell2021noncompliance}, one can point identify $\delta_{f(k)}$ and can in fact identify the effect among individuals who are constant compliers for all factors involved in the interaction or constant compliers across all factors in the experiment.
Partial identification of two-factor interactions among individuals who comply with both factors under weaker assumptions is discussed in Supplementary Materials B.
For this to hold, Assumptions~\ref{assump:mono}, \ref{assump:order_comp}, and~\ref{assump:weak_treat_ex} must hold simultaneously for all factors involved (the same as if constant complier main effects were desired for multiple factors).
This requires additional careful assessment of the assumptions and how this restricts behavior across factors.



\section{Estimation and inference}\label{append:est_inf}
The bounds in the previous sections are unobservable because they are functions of the potential outcomes. We can derive estimators for the bounds with a plug-in estimation approach that uses observed means.
\[
  \overline{Y}^{\text{obs}}(\z) = \frac{1}{N_{\z}} \sum_{i=1}^{N} \mathbb{I}(\bm{Z}_{i} = \z)Y_{i},  \qquad   \overline{D}_{k}^{\text{obs}}(\z) = \frac{1}{N_{\z}} \sum_{i=1}^{N} \mathbb{I}(\bm{Z}_{i} = \z)D_{ik}.
\]
Let $\widehat{\bm{Y}} = (\overline{Y}^{\text{obs}}(\z_{1}), \ldots, \overline{Y}^{\text{obs}}(\z_{J}))$ and let $\widehat{\nu}_{k}(\z_{-k}) = \overline{D}_{k}^{\text{obs}}(\z_{-k}, z_{k}^{+}) - \overline{D}_{k}^{\text{obs}}(\z_{-k}, z_{k}^{-})$. Then, a natural estimator for the bounds from Proposition~\ref{prop:weak_tr_ex} would be
\begin{align*}
       \widehat{\tilde{b}}_k & =  \frac{1}{2^{K-1}\widehat{\nu}_k(\tilde{\z}_{-k})}\left(  \sum_{\z_{-k} \in \mathcal{Z}_{-k}} (\widehat{\nu}_k(\z_{-k}) - \widehat{\nu}_k(\tilde{\z}_{-k}))\right), &\qquad 
    \widehat{\delta}^{\prime}_k & =  \frac{1}{2^{K-1}\widehat{\nu}_k(\tilde{\z}_{-k})} \bm{g}_{k}^{T}\widehat{\bm{Y}}.
\end{align*}
 
It is straightforward to obtain variance estimators for these quantities using the delta method, based on either the superpopulation or finite-population perspectives \citep[see, for example][]{Pashley19, blackwell2021noncompliance}. With these variance estimators, we can use the usual constructions to obtain confidence intervals for each end of the interval, though we can also obtain confidence intervals for the parameter of interest, $\delta_{k}$, rather than the intervals, as described by \cite{ImbMan04}. This latter procedure involves altering the critical values from the confidence intervals to account for how the true parameter cannot simultaneously be close to the upper and lower bound when they are far apart. Thus, when the bounds are sufficiently far apart, this allows us to construct $1-\alpha$ confidence intervals for $\delta_{k}$ using $\alpha$-level critical values on both the upper and lower bound, rather than the usual $\alpha/2$ values.

\end{appendices}
\fi

\end{document}